\documentclass[aps
	,pra
	,letterpaper
    ,twocolumn
	,superscriptaddress
	,amsmath
	,amssymb
	,amsfonts
	,eqsecnum
	,nofootinbib
	,longbibliography
	,nobalancelastpage
	,notitlepage]{revtex4-2}
\usepackage{graphicx}
\usepackage{capt-of}
\usepackage{amsthm}
\usepackage{tabulary}
\usepackage{mathtools}
\usepackage{bm}
\usepackage{braket}
\usepackage{datetime}
\usepackage{stackrel}
\usepackage{stmaryrd}
\usepackage{dsfont}
\usepackage[all]{xy}
\usepackage[hats]{quantum_local}
\usepackage[english]{babel}
\usepackage[normalem]{ulem}

\usetikzlibrary{calc,matrix}

\usepackage{citestyle}

\makeatletter
\@ifundefined{lemma}{\newtheorem{lemma}[theorem]{Lemma}}{}
\@ifundefined{corollary}{\newtheorem{corollary}[theorem]{Corollary}}{}
\makeatother

\begin{document}

\newcommand{\pcc}{_{\oplus^*}}
\newcommand{\bcc}{_{\boxplus^*}}
\newcommand{\sspace}{\mathcal{S}^{\Gamma}_{2n}}
\newcommand{\gammaspace}{\Delta^\Gamma_{2n}}
\newcommand{\fockbargn}{\mc{F}_n}
\newcommand{\symcayn}{\mathrm{Sp}_{2n}^{\Gamma}}
\newcommand{\symcayN}{\mathrm{Sp}_{4n}^{\Gamma}}
	\newcommand{\symplusn}{\mathrm{Sp}_{2n}^+}
	\newcommand{\sympluscayn}{\mathrm{Sp}_{2n}^{+\Gamma}}
	\newcommand{\sympluscayN}{\mathrm{Sp}_{4n}^{+\Gamma}}
	\newcommand{\gchannels}{\mathcal{G}}
	\newcommand{\gchannelsTP}{\mathcal{G}_{\mathrm{TP}}}
	\newcommand{\symplecticcomplexn}{\mathrm{Sp}_{2n}(\complex)}
	\newcommand{\symplecticn}{\mathrm{Sp}_{2n}}
	\newcommand{\symplecticN}{\mathrm{Sp}_{4n}}
	\newcommand{\kmat}{\sigma_z}
\newcommand{\amat}{\mc{A}}
\newcommand{\smat}{\mc{A}}
\newcommand{\matricesn}{M^n(\complex)}
\newcommand{\siegeln}{\Sigma_n}
\newcommand{\siegelN}{\Sigma_{2n}}
\newcommand{\siegeldiskn}{\Delta_n}
\newcommand{\siegeldiskN}{\Delta_{2n}}

\title{Gaussian dynamics in the double Siegel disk}

\author{Giacomo Pantaleoni}
\affiliation{School of Physics \& Astronomy, Monash University, Clayton, VIC 3800, Australia}
\author{Nicolas C. Menicucci}
\affiliation{Centre for Quantum Computation and Communication Technology, School of Science, RMIT University, Melbourne, VIC 3000, Australia}

\date{\today~at~\currenttime}

\begin{abstract}
		We show that deterministic multimode Gaussian channels admit a symmetric-space description. Passing from the n-mode Siegel disk to a doubled version of that space lets general Gaussian dynamics act by a linear-fractional (Mobius) transformation on a single matrix parameter. This doubled disk naturally parametrizes Gaussian kernels in the Fock-Bargmann representation, and contains an explicit physical subset corresponding to valid mixed Gaussian states. Starting from the standard X,Y parametrization of a deterministic Gaussian channel, we construct a normalized oscillator-semigroup element whose fractional action reproduces the channel update on that subset; Gaussian unitaries appear as the symplectic, isometric special case. This gives a bridge between covariance-matrix channel theory and the adjacency-matrix or symmetric-space picture, preserves a simple composition law given by matrix multiplication of the acting blocks, and suggests a direct route to graphical update rules beyond pure states.
\end{abstract}
\maketitle

	\section{Introduction}
	The set of all Gaussian pure states can be viewed as a space of symmetric matrices~\cite{simon_gaussian_1988}. Two such symmetric spaces are the Siegel upper half-plane and the Siegel disk~\cite{folland_harmonic_1989,menicucci_graphical_2011,gabay_passive_2016}. Gaussian dynamics of pure states is then described by tracking the evolution of these symmetric matrices under the action of linear fractional transformations---special transformations built from the blocks of a symplectic matrix that parametrizes the Gaussian unitary (via the metaplectic representation)~\cite{simon_gaussian_1988}. The half-plane and the disk (symmetric spaces) are related by the same linear transformation that maps quadrature operators into bosonic annihilation and creation operators: The half-plane describes a state from a field-quadrature point of view, while the disk description is more closely related to Fock space.

Symmetric spaces are of interest to at least two separate research areas. On one hand, the mathematical physics literature~\cite{folland_harmonic_1989,hilgert_note_1989,howe_oscillator_1988,fabec_holomorphic_2006,derezinski_quantization_2020} has considered the action of Gaussian transformations on symmetric spaces as an application of the metaplectic representation and the oscillator semigroup, which generalizes the former to non-unitary (typically non-deterministic) single-Kraus dynamics. Symmetric spaces emerge when studying the representation of these (semi-)groups on Gaussian states. Typical questions in this context relate to complications introduced by the semigroup structure of non-unitary Gaussian transformations on \emph{pure} Gaussian states~\cite{howe_oscillator_1988,folland_harmonic_1989}. Namely, the composition law, relationship to phase-space representations of non-unitary operations, and exploration of the (semi)-group structure (including subsemigroups) of the oscillator semigroup are all themes investigated by this line of research.

	On the other hand, the quantum optics community~\cite{simon_gaussian_1988,menicucci_graphical_2011,gabay_passive_2016} has been using symmetric spaces to describe specific quantum optical systems and continuous-variable protocols, especially in quantum information.
	Representative quantum-information work exploiting Gaussian graph/adjacency-matrix descriptions of multimode Gaussian states includes Refs.~\cite{alexander_universal_2018,walshe_streamlined_2021,xin_qumode_2023,descamps_noisy_2024}. When the dynamics of interest is fully Gaussian and restricted to pure Gaussian states, one might be tempted to move to a phase-space description such as the Wigner picture and rely on the fact that the evolution reduces to matrix operations on the covariance matrix. This is useful, but not entirely efficient. Under Gaussian unitaries, we do not need the full covariance matrix of a pure state: its symmetric matrix suffices, and unitaries act on it by simple matrix operations. In particular, the composition law is still given by matrix multiplication.

They are also expressive: as states are identified with symmetric matrices, it is possible to interpret them as adjacency matrices corresponding to graphs, which can be represented graphically in a framework~\cite{menicucci_graphical_2011} that essentially generalizes discrete-variable graph states~\cite{vandennest_graphical_2004}. The formalism is used to construct protocols for universal quantum computation that leverage Gaussian states~\cite{walshe_streamlined_2021,alexander_universal_2018}. It is a flexible tool available to continuous-variable quantum information. References~\cite{xin_qumode_2023,descamps_noisy_2024,nokkala_complex_2024} provide some examples of its use in current research.

We note in passing that discrete-variable graph states (more in general, stabilizer states) can be cast in the same framework, by considering an appropriately discretized version of the symplectic group. For example, the discrete version of these symmetric spaces and linear fractional transformations on them was used to describe the evolution of graph states~\cite{vandennest_graphical_2004,vandennest_efficient_2004}. Symmetric spaces can therefore provide precise analogies between stabilizer states and Gaussian states, and the Clifford group and the Gaussian group.

	Despite the usefulness and the mathematical elegance of Siegel spaces, a prominent shortcoming of the formalism is the inability to describe mixed Gaussian states or general (deterministic, CPTP) Gaussian channels---at best it captures single-Kraus (non-deterministic, trace-nonincreasing) Gaussian operations on pure states (implemented by the oscillator semigroup), although the framework is usually presented in the even more restrictive regime of fully unitary dynamics alone.

This is the problem that we solve in the present work by providing a generalization of the formalism that is at least as powerful as the traditional covariance matrix-based approach to non-unitary Gaussian dynamics. In doing so, we will also take the opportunity to stress that, even without the mixed-state extension, symmetric spaces already allow one to consider simple Gaussian non-unitary dynamics (albeit not the whole picture) with no need for extension. This latter observation is not novel, although rarely fully spelled out in the continuous-variable quantum information literature. To the author's knowledge, prior to this work, the task of connecting the dots in the literature faced by a quantum information theorist interested in representing a non-unitary single-Kraus Gaussian operation, was somewhat non-trivial.

After reviewing the formalism of the upper-half plane (arguably the more popular symmetric space in the context of quantum optics), we choose to focus our attention on the Siegel disk instead~\cite{menicucci_graphical_2011,gabay_passive_2016}. As we will show, the Siegel disk is related to the Fock-Bargmann kernel elements~\cite{bargmann_hilbert_1961,bargmann_representations_1962}, which are Gaussian for Gaussian states and operations. The Fock-Bargmann representation has been the object of recent literature under various names, and it is essential in the study of, in fact, \emph{non}-Gaussian states~\cite{chabaud_holomorphic_2021,walschaers_nongaussian_2021,miatto_fast_2020,yao_recursive_2022,chabaud_resources_2023} and Gaussian boson sampling~\cite{hamilton_gaussian_2017,arrazola_using_2018, kruse_detailed_2019,bulmer_boundary_2022,grier_complexity_2022, deng_solving_2023}. The Fock-Bargmann representation of states is closely related to the Fock space representation, to the extent that the Fock-Bargmann basis can be interpreted as a realization of number states on a Hilbert space of holomorphic wavefunctions. As such, the Fock-Bargmann Hilbert space is sometimes simply referred to as ``Fock space''\cite{folland_harmonic_1989,zhu_analysis_2012}.

This paper is structured as follows. In \cref{sec:definitions}, we establish the covariance-matrix conventions and symplectic notation used throughout. In \cref{sec:fractional-transformations}, we review the pure-state symmetric-space picture and the linear fractional transformations that encode Gaussian dynamics, and we connect the upper-half-plane and disk pictures via the Cayley transform. In \cref{sec:siegel-disk}, we develop the mixed-state and channel extension in the disk representation, culminating in the channel-level disk fractional update rule. We conclude in \cref{sec:conclusion-outlook} with an outlook and open problems.

\section{Basic definitions and conventions}\label{sec:definitions}

We work with \(n\)-mode systems. We write \(\hat{q}\) (resp.\ \(\hat{p}\)) for the column vector of position (resp.\ momentum) operators, and we stack them into the \(2n\)-vector
\(
\hat{r} \coloneqq (\hat{q}^{T},\hat{p}^{T})^{T}
\).
We fix the displacement-operator convention
\(
D(\xi)\coloneqq e^{i \xi^{T}\hat{r}}
\),
and, unless stated otherwise, we restrict attention to zero-mean states and treat displacements as an irrelevant gauge.

\begin{definition}[Real symplectic group]\label{def:symplectic-group}
	Let \(n\ge 1\) and define the standard symplectic form
	\begin{align}
		\Omega \coloneqq
		\begin{pmatrix}
			0 & I_{n} \\
			- I_{n} & 0
		\end{pmatrix}
		\,.
	\end{align}
	The real symplectic group \(\symplecticn\) is the set of all \(2n\times 2n\) real matrices \(S_{r}\) such that
	\begin{align}\label{eq:symplecticndef}
		S_{r}^{T} \Omega S_{r} = \Omega
		\,.
	\end{align}
\end{definition}

Pauli matrices will be ubiquitous throughout, although their size will depend on context. Thus, we define
\begin{definition}[Lifted Pauli matrices]\label{def:lifted-paulis}
	Let \(\sigma_{x},\sigma_{y},\sigma_{z}\in M_{2}(\complex)\) denote the standard Pauli matrices. For any block size \(m\ge 1\), define their lifted versions by
	\begin{align}
		\sigma_{\mu}^{[m]} \coloneqq \sigma_{\mu}\otimes I_{m}\in M_{2m}(\complex),
		\qquad
		\mu\in\{x,y,z\}
		\,.
	\end{align}
	Equivalently, the lifted Pauli matrices are the \(2m\times 2m\) block matrices
	\begin{align}
	\label{eq:sigmablockdef}
		\sigma_{x}^{[m]}
		 &\coloneqq
		\begin{pmatrix}
			0    & I_{m} \\
			I_{m} & 0
		\end{pmatrix}
		\,,
		\\
		\sigma_{y}^{[m]}
		 &\coloneqq
		\begin{pmatrix}
			0     & - i I_{m} \\
			i I_{m} & 0
		\end{pmatrix}
		\,,
		\\
		\sigma_{z}^{[m]}
		 &\coloneqq
		\begin{pmatrix}
			I_{m} & 0 \\
			0     & - I_{m}
		\end{pmatrix}
		\,.
	\end{align}
	Throughout the manuscript, the relevant value of \(m\) can always be inferred from context, and we will often omit the superscript \([m]\).
\end{definition}

To further suppress clutter, in block-matrix expressions, a scalar entry \(a\in\complex\) is understood as \(a I_{m}\), where \(m\) is the block size implied by context. This convention also includes simply \(1\) in place of an appropriately sized identity \(I_m\).

When working with \(2n\times 2n\) phase-space matrices, we will often identify the symplectic form \(\Omega\) with a lifted Pauli matrix by writing
\begin{align}
	\sigma_{y}^{[n]} = - i \Omega
	\,.
\end{align}
When no ambiguity can arise, we will suppress the superscript and write \(\sigma_{y}\) in place of \(\sigma_{y}^{[n]}\).
For equalities such as \cref{eq:symplecticndef}, one may equivalently work with \(\Omega\) or \(\sigma_{y}^{[n]}\). However, the factor of \(-i\) is not a cosmetic choice: it makes \(\sigma_{y}^{[n]}\) Hermitian, which is essential once we consider semigroup conditions and matrix inequalities later in the manuscript.

We choose this convention because each Pauli matrix will make an appearance throughout this article, playing different roles, especially when giving standard presentations for (semi-)groups. In addition, it allows readers to track how some of the changes of coordinates we perform manifest as a rotation of Pauli matrices among themselves.

\begin{definition}[Gaussian state]\label{def:gaussian-state}
	A zero-mean \(n\)-mode Gaussian state is fully described by its \(2n\times 2n\) covariance matrix
	\begin{align}
		\Sigma_{jk} \coloneqq \frac{1}{2}\tr\!\pqty{\rho\{\hat r_{j},\hat r_{k}\}}
		\,.
	\end{align}
	By Williamson's decomposition, \(\Sigma\) can be written as
	\begin{align}
		\label{eq:covariance_matrix}
		\Sigma
		=
		\frac{1}{2}
		S_{r}
		D
		S_{r}^{T}
		\,,
	\end{align}
	where \(S_{r}\in\symplecticn\) and
	\(
	D=
	\begin{psmallmatrix}
		\nu & 0\\
		0 & \nu
	\end{psmallmatrix}
	\),
	with \(\nu=\mathrm{diag}(\nu_{1},\dots,\nu_{n})\) and \(\nu_{i}\ge 1\)~\cite{simon_quantumnoise_1994,serafini_quantum_2017,adesso_continuous_2014}.
\end{definition}

The uncertainty principle implies \(\nu_{i}\ge 1\). When \(\nu_{i}=1\) for all \(i\), the state is pure.
A Gaussian unitary is a unitary that maps every Gaussian state to a Gaussian state; by the Stone--von Neumann theorem, these are exactly the metaplectic unitaries (up to displacements, which we suppress throughout). More generally, a Gaussian \emph{operation} is a completely positive, trace-nonincreasing map that preserves the set of Gaussian states; its trace-preserving subclass are the Gaussian channels. Following standard terminology (e.g.\ Serafini~\cite{serafini_quantum_2017}), we will sometimes refer to trace-preserving Gaussian channels as \emph{deterministic} and trace-nonincreasing Gaussian operations as \emph{non-deterministic}.

To describe how Gaussian states transform under symplectic operations, we use the metaplectic representation.

\begin{definition}[Metaplectic representation and group]\label{def:metaplectic}
	For each \(S_{r}\in\symplecticn\), the (projective) metaplectic representation assigns a unitary \(U_{S_{r}}\) (unique up to an overall sign) such that
	\begin{align}\label{eq:metaplecticdefint}
		D(\xi) U_{S_{r}} = U_{S_{r}} D(S_{r}^{T} \xi)
		\,,
	\end{align}
	for all \(2n\)-dimensional real vectors \(\xi\).
	The set of all such unitaries (under composition) forms the metaplectic group, a double cover of \(\symplecticn\).
\end{definition}

With this convention, if a state \(\rho\) evolves in the Schr\"odinger picture as \(\rho' = U_{S_{r}} \rho U^{\dag}_{S_{r}}\), then its covariance matrix transforms as
\begin{align}
	\label{eq:sigmaprime}
	\Sigma' = S_{r} \Sigma S_{r}^{T}
	\,.
\end{align}
The composition law for symplectic operations is matrix multiplication. The covariance matrix \(\Sigma\) appears in the characteristic function (directly) and in the Wigner function (through \(\Sigma^{-1}\)).

For later use (particularly in \cref{sec:fractional-transformations}), we collect the block structure of symplectic matrices. A \emph{real} symplectic matrix \(S_{r}\) can be decomposed into four \(n\times n\) blocks:
\begin{align}
	\label{eq:ABCD}
	S_{r}
	 & =
	\begin{pmatrix}
		A & B \\
		C & D
	\end{pmatrix}
	\,,
\end{align}
and condition \cref{eq:symplecticndef} then reads
\begin{align}
	A^{T} C
	=
	C^{T} A
	\,,\quad
	B^{T} D
	=
	D^{T} B
	\,,\quad
	A^{T} D
	-
	C^{T} B
	=
	1
	\label{eq:sympblocks}
	\,.
\end{align}

As anticipated, we will prefer the Siegel disk to the upper half-plane. When working in the disk picture, it is more convenient to use a complex version of \(\Sigma\) by conjugating with a fixed matrix that we define as

\begin{definition}[Cayley matrix]\label{def:cayley-matrix}
	The \emph{Cayley matrix} is the \(2\times 2\) matrix
	\begin{align}
			\Gamma
		\coloneqq
		\frac{1}{\sqrt{2}}
		\begin{pmatrix}
			1 & - i \\
			1 & i   \\
		\end{pmatrix}
		\,.
	\end{align}
\end{definition}

Conjugation by \(\Gamma\) is the operation that transports real phase-space objects (symplectic matrices, covariance matrices) into the complex form natural to the disk picture. We give this operation a name for convenience.

\begin{definition}[Adjoint-by-Cayley (ABC)]\label{def:abc}
	Given a real \(2n\times 2n\) matrix \(X_{r}\), we define its \emph{adjoint-by-Cayley} (ABC) counterpart as
	\begin{align}
		\label{eq:cayley_matrix}
			X
			\coloneqq
			\Gamma
			X_{r}
			\Gamma^\dag
		\,.
	\end{align}
\end{definition}
so that the \emph{complex covariance matrix}, \(\sigma\), is defined as the ABC real covariance matrix
\begin{align}
	\label{eq:complex_covariance_matrix}
		\sigma
	=
	\Gamma \Sigma \Gamma^{\dag}
	=
	\frac{1}{2}
		S
	(\nu \oplus \nu)
		S^{\dag}
	\,,
\end{align}
	where $S = \Gamma S_{r} \Gamma^{\dag}$ is an element of the ABC symplectic group \(\symcayn \cong \mathrm{Sp}(n, \mathbb{R})\) (the isomorphism being \(\Gamma\)), and
\( \Gamma (\nu \oplus \nu) \Gamma^{\dag} = \nu \oplus \nu \). The diagonal matrix \(\nu\) must satisfy \(\nu \ge 1\), and can be regarded as a set of temperatures. Note that the definition of \( \sigma \) is also a matter of conventions: Other works consider its complex conjugate instead \cite{adesso_continuous_2014,yao_recursive_2022}, since they work with the complex conjugate of \( \Gamma \).
The matrix $\sigma$ transforms under the Gaussian unitary identified by \(S\) as
\begin{align}
	\label{eq:covariance_transformation}
		\sigma' & =  S \sigma S^{\dag} \,.
\end{align}
Note that the Cayley matrix induces the Fock-metaplectic representation as an intertwiner between displacement operators expressed as exponentials of creation and annihilation operators. One can see this by rewriting~\cref{eq:metaplecticdefint} in terms of bosonic annihilation and creation operators. The Cayley matrix \(\Gamma\) allows one to pass from position/momentum to annihilation/creation, as well as from a real symplectic matricex \(S_{r}\) to an ABC \(S\) matrix.

\section{Pure Gaussian Dynamics in Symmetric Spaces}
\subsection{Kinematics}
	A pure Gaussian state is a state whose wavefunction is a multivariate Gaussian function. In the position representation, where a position eigenstate is \(\ket{x}\), Gaussian states are parametrized with an adjacency matrix as~\cite{simon_gaussian_1988,menicucci_graphical_2011}
\begin{align}
	\label{eq:gaussian_pure_pos}
		\bra{x} \ket{Z}
	=
		\pqty{\frac{\det(\operatorname{Im} Z)}{\pi^{n}}}^{\frac{1}{4}}
		e^{\frac{i}{2} x^T Z x}
	\,
\end{align}
	where \(Z\) is a symmetric matrix with positive definite imaginary part. (Symmetric matrices suffice to describe all Gaussian states, as the product $x^T Z_0 x$ kills off the antisymmetric part of any matrix $Z_0$.) Positive definitess of the imaginary part ensures that the states are square integrable functions. The set of matrices satisfying these conditions is called the \emph{Siegel upper half-plane}:
\begin{definition}[Siegel upper half-plane]\label{def:siegel-uhp}
	\begin{align}\label{eq:siegelplanedef}
		\siegeln =
		\{
			Z \in \matricesn \big|
			Z = Z^{T}
		,\
			\text{Im} {Z} > 0
		\}
		\,,
	\end{align}
\end{definition}
which therefore we identify with the set of pure gaussian states themselves.

\begin{center}
	\begin{minipage}{\linewidth}
	\centering
	\begin{tikzpicture}[x=1.1cm,y=1.1cm,>=stealth,every node/.style={font=\scriptsize}]
		\node[font=\scriptsize] at (-1,2.2) {\(\Sigma_{1}\)};
		\fill[gray!8] (-1.45,0) rectangle (1.45,1.85);
		\draw[thick,dashed] (-1.45,0) -- (1.45,0);
		\draw[thick,->] (0,0) ++(90:1.05) arc (90:40:1.05)
			node[midway,above] {\(S_r\)};
		\fill (0,1.05) circle (1.2pt) node[above left] {\(Z_{0}=i\)};
		\fill ({cos(40)*1.05},{sin(40)*1.05}) circle (1.2pt) node[right] {\(Z_{1}\)};
		\draw[->] (-1.55,0) -- (1.55,0) node[below] {\(\reals\)};
		\draw[->] (0,-0.3) -- (0,1.95) node[left] {\(i\reals\)};
	\end{tikzpicture}
		\captionof{figure}{Upper half-plane intuition (\(n=1\)). The Siegel upper half-plane is \(\siegeln=\{Z=Z^{T},\ \operatorname{Im}Z>0\}\); for one mode it reduces to \(\{z\in\complex\mid \operatorname{Im}z>0\}\). The vacuum corresponds to \(Z_{0}=i\) (i.e.\ \(iI_n\) in general). A Gaussian unitary \(S_r\in\symplecticn\) acts by a M\"obius-type fractional transformation \(Z\mapsto\phi_{S_r}(Z)\), which preserves \(\operatorname{Im}Z>0\).}
	\label{fig:uhp-intuition}
	\end{minipage}
\end{center}

The quantities we deal with in this paper are more closely related to the Fock-Bargmann \cite{bargmann_representations_1962} wavefunction of Gaussian states, also known as stellar function \cite{chabaud_stellar_2020,chabaud_resources_2023}, rather than the wavefunction in position representation. It is useful to define it as an inner product with coherent states as it makes the connection with the Husimi function, which we require to establish for our purposes, more transparent. (A more standard definition can be found in~\cite{folland_harmonic_1989}.) If \(\ket{n}\) is a Fock state, and a coherent state is
\begin{align}\label{eq:coherent}
	\ket{\zeta} = e^{-\half \abs{\zeta}^{2}} \sum_{n=0}^{\infty} \frac{\zeta^{n}} {\sqrt{n!}} \ket{n}
	\,,
\end{align}
then the Fock-Bargmann wavefunction is obtained by projecting a state onto a coherent state and removing the Gaussian envelope.

\begin{definition}[Fock-Bargmann wavefunction]\label{def:fock-bargmann}
	The Fock-Bargmann wavefunction is the entire function \(\psi\) that can be constructed from the Fock-state expansion
	\begin{align}
		\label{eq:fock_purestate}
		\ket{\psi} & =
		\sum_{n} \psi_n \ket{n}
	\end{align}
	from the inner product with a coherent state~\cite{ali_coherent_2014}:
	\begin{align}\label{eq:fockbdef}
		\psi(\zeta)
		\coloneqq
		e^{\half \abs{\zeta}^{2}} \bra{\zeta^{*}} \ket{\psi}
		=
		\sum_{n} \frac{\zeta^n}{\sqrt{n!}} \psi_n
		\,.
	\end{align}
\end{definition}
The Fock-Bargmann wavefunction is thus an expansion with respect to the Fock coefficients,
$\psi(\zeta) = \sum_{n} \psi_{n} u_{n}(\zeta)$,
in terms a basis of entire functions
$u_{n} (\zeta) = \frac{\zeta^n}{\sqrt{n!}} $.
The scalar product of two wavefunctions $\psi$ and $\phi$, in this Hilbert space, must be taken with respect to the measure \(d\mu(\zeta)\), as
\begin{align}\label{eq:fbinnerproduct}
	\braket{\psi} {\phi} & = \int d \mu(\zeta) \psi^{*} (\zeta) \phi(\zeta)
	\\
	d \mu                & = \pi^{-n} e^{ -\abs{\zeta}^{2} } d\zeta
	\,,
\end{align}
where $d\zeta$ is the $2n$-dimensional Lebesgue measure in $\complex^{n}$.
With respect to the inner product of this Hilbert space, $\fockbargn$, the eigenstates are normalized in the sense that \(\bra{\zeta} \ket{\omega} = e^{\zeta \omega^{*}}\), so that
\begin{align}\label{eq:fbdeltaexpansion}
	\psi(\zeta) = \int d \mu(\omega) e^{\zeta \omega^{*}} \psi(\omega)
	\,.
\end{align}
The isometry that maps \(L^{2}(\reals^{n})\) into \(\fockbargn\) is the Bargmann transform\cite{bargmann_hilbert_1961,bargmann_representations_1962}. We will not need its explicit form, but suffices it to say that it maps Gaussian wavefunctions in position representation to Gaussian wavefunctions in Fock-Bargmann representation~\cite{bargmann_hilbert_1961,hilgert_note_1989}, which are parametrized as 
\begin{align}
	\label{eq:gaussian_fock-bargmann}
		\bra{\zeta} \ket{K}
	=
	\det\pqty{1 - K K^{*}}^{\frac{1}{4}}
		e^{ \frac 1 2 \zeta^{T} K \zeta }
	\,,
\end{align}
with \(\zeta\) \(n\)-dimensional complex variable.
	The requirement that the above is an entire function means that $\norm {K} < 1$, where $\norm{K}$ is the operator norm of $K$. This is equivalent to
	$K^{*} K < 1$.
In addition, \( K^{T} = K \). The set of matrices satisfying these conditions is the Siegel disk.

\begin{definition}[Siegel disk]\label{def:siegel-disk}
	The \emph{Siegel disk} is
	\begin{align}\label{eq:siegeldiskdef}
		\siegeldiskn =
		\{
			K \in \matricesn \big|
			K = K^{T}
		\,,
			K^{*} K < 1
		\}
		\,.
	\end{align}
\end{definition}

We can think of this space as a disk of matrices of radius one, a generalization of the Poincar\'e disk.
Since the Bargmann transform is an isometry, or, alternatively, since the Cayley transform is a bijection, the disk is equivalent to the half-plane and its elements represent all Gaussian states.

\subsection{Dynamics}\label{sec:fractional-transformations}

We briefly go back to the half-plane, before abandoning it in favor of the disk.
The adjacency matrix \(Z\) of a pure Gaussian state transforms under a symplectic operation via a matrix-valued generalization of a M\"obius transformation, called a \emph{linear fractional transformation}.

	\begin{definition}[Linear fractional transformation]\label{def:fractional-transformation}
	A \emph{linear fractional transformation} with respect to a \(2n\times2n\)-dimensional matrix \(T\) written in \(n\times n\) block form as
	\begin{align}
		T =
		\begin{pmatrix}
			A & B \\
			C & D
		\end{pmatrix}
		\,,
	\end{align}
		is the map~\cite{folland_harmonic_1989,simon_gaussian_1988,menicucci_graphical_2011}
	\begin{align}
		\label{eq:linearfractionalreal}
		\phi_{T}(Z)
		 & =
		( D Z + C )
		( B Z  + A )^{-1}
		\,.
	\end{align}
	For a generic block matrix \(T\), \(\phi_T\) is a partially-defined rational map (it is only defined where the inverse exists). In this manuscript, \(T\) will always belong to specific groups/semigroups (symplectic or Siegel-domain preserving) chosen precisely so that \(\phi_T\) preserves the relevant Siegel domain (half-plane or disk), making it a well-defined dynamical update on the corresponding state coordinate.
\end{definition}

	The two key properties of fractional transformations induced by symplectic matrices are that they preserve the Siegel upper half-plane and act transitively on it. These are classic results (see, e.g., Ref.~\cite{folland_harmonic_1989}), but later on, we will still sketch out the proofs of the equivalent statements in the disk.

\begin{theorem}[Symplectic action preserves the Siegel upper half-plane]\label{thm:symplectic-preserves-uhp}
	If \(S_{r}\in\symplecticn\), then \(\phi_{S_{r}}(\siegeln)\subseteq\siegeln\). Equivalently, \(\phi_{S_{r}}\) maps pure Gaussian states to pure Gaussian states in the upper half-plane representation. See Ref.~\cite{folland_harmonic_1989}.
\end{theorem}

\begin{theorem}[Transitivity/onto property]\label{thm:uhp-transitive}
	For any \(Z,Z'\in\siegeln\), there exists \(S_{r}\in\symplecticn\) such that \(Z'=\phi_{S_{r}}(Z)\). In particular, the symplectic action on \(\siegeln\) is transitive (equivalently, fractional transformations with respect to symplectic matrices are onto). See Ref.~\cite{folland_harmonic_1989}.
\end{theorem}

More specifically, there exists a linear fractional transformation that maps the vacuum $i$ into any $Z \in \siegeln$,
	\begin{align}
		\label{eq:fromvacuumz}
		Z
		=
	\pqty{D i + C}
		\pqty{B i + A}^{-1}
			\,.
		\end{align}
Here \(i\) denotes \(iI_{n}\): inserting \(Z=iI_{n}\) into \cref{eq:gaussian_pure_pos} gives \(\braket{x}{i}=\pi^{-n/4}e^{-\frac12 x^{T}x}\), the standard vacuum wavefunction (equivalently, \(\Sigma=\tfrac12 I_{2n}\) in the real-quadrature convention of \cref{eq:covariance_matrix}).

	We will sometimes use an alternative description of fractional transformations in stacked notation~\cite{menicucci_graphical_2011}. Concretely, we stack two \(n \times n\) matrices (\(P\) and \(Q\)) into a \(2n \times n\) matrix \(\pqty{P^{T},Q^{T}}^{T}\), but we regard two such stacks as the same object if they differ by a change of representative (right multiplication) by an invertible \(n\times n\) matrix. Thus, we write
\begin{align}
	\begin{bmatrix}
		P \\ Q
	\end{bmatrix}
	\coloneqq
	\Bqty{
		\begin{pmatrix}
			P H
			\\
			Q H
		\end{pmatrix}
		\ \Big|\ 
		H\in \mathrm{GL}(n)
	}
	\,,
	\label{eq:stackednotation}
\end{align}
	and any choice of \(H\) represents the same stacked object; this is the standard projective/homogeneous-coordinate description of the relevant Grassmannian quotient (see, e.g., Ref.~\cite{freitas_revisiting_2004}).
In this notation, we can think of a fractional transformation as a matrix multiplication from left of \((P^{T}, Q^{T})^T\) by a \(2n \times 2n \) matrix \(S_{r}\), modulo matrix multiplication from right. Namely,
\begin{align}
	\label{eq:LABEL}
	\begin{pmatrix}
		A
		 &
		B
		\\
		C
		 &
		D
	\end{pmatrix}
	\begin{bmatrix}
		1 \\
		Z \\
	\end{bmatrix}
	=
	\begin{bmatrix}
		1
		\\
		\pqty{D Z + C}
		\pqty{B Z + A}^{-1}
	\end{bmatrix}
	\,
\end{align}
maps the representative $\bqty{1, Z}^{T}$ to the representative whose \(Z\) matrix undergoes the fractional transformation~\cref{eq:linearfractionalreal}. Namely,
\(
S_r\begin{bmatrix}1\\ Z\end{bmatrix}
=\begin{bmatrix}A+BZ\\ C+DZ\end{bmatrix}
=\begin{bmatrix}1\\ \phi_{S_r}(Z)\end{bmatrix}(BZ+A)
\),
and right-multiplication by \(BZ+A\in\mathrm{GL}(n)\) is exactly the representative-change in \cref{eq:stackednotation} (this is also where invertibility enters). In the stacked picture, it is clear that the composition law for fractional transformations is matrix multiplication. We state this as a Lemma,
\begin{lemma}[Composition law for fractional transformations]\label{lem:fractional-composition-law}
	Let \(S_1,S_2\in\symplecticcomplexn\). Whenever both sides are defined (i.e., the relevant block inverses exist), the induced linear fractional transformations satisfy
	\begin{align}
		\label{eq:compositionlaw}
		\phi_{S_2}
		\circ
		\phi_{S_1}
		=
		\phi_{S_2 S_1}
		\,.
	\end{align}
\end{lemma}
This is well known, but sketching out the proof presents a good opportunity to illustrate how to work with stacked notation. 
	\begin{proof}[Proof (Sketch)]
		In stacked notation, \(Z\) is represented by \(\begin{bmatrix}1\\ Z\end{bmatrix}\) up to right multiplication by \(\mathrm{GL}(n)\). If \(S_1=\begin{psmallmatrix}A_1&B_1\\ C_1&D_1\end{psmallmatrix}\), then
		\begin{align}
			S_1\begin{bmatrix}1\\ Z\end{bmatrix}
		=
		\begin{bmatrix}A_1+B_1Z\\ C_1+D_1Z\end{bmatrix}
		=
		\begin{bmatrix}1\\ \phi_{S_1}(Z)\end{bmatrix}(B_1Z+A_1),
	\end{align}
		where right multiplication by \((B_1Z+A_1)\in \mathrm{GL}(n)\) is a representative change. Applying \(S_2=\begin{psmallmatrix}A_2&B_2\\ C_2&D_2\end{psmallmatrix}\) and using associativity,
		\begin{align}
			S_2S_1\begin{bmatrix}1\\ Z\end{bmatrix}
			&\sim
			S_2\begin{bmatrix}1\\ \phi_{S_1}(Z)\end{bmatrix}
			\\
			&=
			\begin{bmatrix}1\\ \phi_{S_2}(\phi_{S_1}(Z))\end{bmatrix}
			\pqty{B_2\phi_{S_1}(Z)+A_2}
			\,,
		\end{align}
			so \(S_2S_1\) induces the same fractional update as \(\phi_{S_2}\circ\phi_{S_1}\), i.e. \(\phi_{S_2}\circ\phi_{S_1}=\phi_{S_2S_1}\) wherever all inverses exist.
		\end{proof}
Using stacked notation is a clever shortcut devised by Siegel to simplify the proof of the two basic theorems presented above.

The relevance of the Siegel upper half-plane to physics is explained by the relationship between the action of the metaplectic representation on Gaussian states parametrized in \(\siegeln\). If $S_{r}$ is a real symplectic matrix, then states evolve according the metaplectic representation \(U_{S}\) (which is unitary) as
\begin{align} \label{eq:fraconpositionrep}
	\op U_{S_r} \ket{Z} & = \ket{Z' = \phi_{S_r} (Z)}
	\,.
\end{align}
This can be verified by Gaussian integration and the metaplectic representation with our conventions~\cite{folland_harmonic_1989,derezinski_quantization_2020}.

Gaussian unitaries are represented by the metaplectic group \(\mathrm{Mp}(2n,\reals)\), which double-covers \(\symplecticn\). Accordingly, a given symplectic matrix \(S_r\) has two metaplectic lifts whose unitaries differ by a global sign, \(\pm U_{S_r}\), which is physically irrelevant. Separately, the fractional action on \(Z\) depends only on \(S_r\) modulo \(\{\pm 1\}\) (since \(\phi_{S_r}=\phi_{-S_r}\)), so the induced update \(Z\mapsto\phi_{S_r}(Z)\) is well-defined on the symmetric-space coordinate.

The equality \(\phi_{S_r}=\phi_{-S_r}\) is a feature of the symmetric-space coordinate: the central element \(-I\in\symplecticn\) acts trivially on the Siegel upper half-plane, so the fractional action factors through \(\symplecticn/\{\pm I\}\).
By contrast, the metaplectic ``\(\pm\)'' refers to the fact that a \emph{fixed} \(S_r\in\symplecticn\) has two metaplectic lifts whose unitaries differ by an overall phase (in particular, a sign). These two sources of ``\(\pm\)'' live in different places and should not be conflated.

\subsubsection{Non-unitary dynamics in the Siegel upper half-plane}
Gaussian unitaries are not the most general physically acceptable transformations that map Gaussian pure states to Gaussian pure states: there are also single-Kraus (typically non-deterministic, trace-nonincreasing) Gaussian operations. The corresponding set of matrices acting on \(\siegeln\) is the oscillator semigroup, which we will realize below as a Siegel-domain preserving semigroup of fractional transformations. To specify this set of transformations precisely, we must first define

\begin{definition}[Complex symplectic group]\label{def:complex-symplectic-group}
	The complex symplectic group \(\symplecticcomplexn\) is the set of all \(2n\times 2n\) complex matrices \(S_{c}\) satisfying
	\begin{align}\label{eq:complexsymplecticdef}
		S_{c}^{T} \sigma_{y} S_{c} = \sigma_{y}
		\,.
	\end{align}
	Equivalently, since \(\sigma_{y}^{[n]}=-i\Omega\), this is the usual condition \(S_{c}^{T}\Omega S_{c}=\Omega\).
\end{definition}

To show that a subset of this group preserves the upper half-plane, we leverage the fact that the complex symplectic group preserves symmetry.
Symmetry of \(Z\) has a convenient formulation in stacked notation:
\begin{align}
	\label{eq:zsiegelreq1}
	\begin{bmatrix}
		1 \\
		Z
	\end{bmatrix}^{T}
	\sigma_{y}
	\begin{bmatrix}
		1 \\
		Z
	\end{bmatrix}
	=
	0
	\,,
\end{align}
which is verified by expanding \cref{eq:zsiegelreq1} and noting that it gives a multiple of \(Z^{T}-Z\). Thus, 

\begin{lemma}[Complex symplectic action preserves symmetry]\label{thm:complex-symplectic-preserves-symmetry}
	Write \(S_{c}\in\symplecticcomplexn\) in \(n\times n\) blocks as \(S_{c}=\begin{psmallmatrix}A&B\\ C&D\end{psmallmatrix}\). If \(Z\in M_{n}(\complex)\) is symmetric and \(BZ+A\) is invertible (so that \(\phi_{S_{c}}(Z)\) is well-defined), then \(Z' \coloneqq \phi_{S_{c}}(Z)\) is symmetric.
\end{lemma}
\begin{proof}[Proof (Sketch)]
		Using stacked notation, symmetry of \(Z\) is equivalent to the isotropy condition \cref{eq:zsiegelreq1}.
		This isotropy is preserved under the stacked action of \(S_{c}\) because \(S_{c}^{T}\sigma_{y}S_{c}=\sigma_{y}\), and because changing representatives by right multiplication does not affect whether the quadratic form vanishes.
	\end{proof}

However, the imaginary part of \(Z'\) in this case is not necessarily greater than zero (note that this was always true if \(S \in \symplecticn\), i.e., for Gaussian unitaries). We can give a simple example of a complex symplectic matrix that induces a fractional transformation mapping half-plane elements to other half-plane elements. Consider the matrix
\begin{align}\label{eq:sepsilon}
	S_{\epsilon} =
	\begin{pmatrix}
		1          & 0 \\
		i \epsilon & 1
	\end{pmatrix}
	\,,
\end{align}
which is a complex symplectic matrix associated with the non-unitary evolution under the Gaussian Kraus operator $n(\epsilon) e^{-\epsilon \q^{2}}$ (multiplication by a Gaussian, with \(n(\epsilon)\) such that the final state is normalized). This is a non-unitary operation that damps a wavefunction and renormalizes it as $n(\epsilon) e^{-\epsilon \q^{2}} \psi(x) = n(\epsilon) e^{-\epsilon x^{2}} \psi(x)$. On Gaussian states, the operation gives
$e^{-\epsilon \q^{2}}\ket{Z} = \ket{\phi_{S_{\epsilon}}(Z)} = \ket{Z + i\epsilon}$. Clearly, if $\epsilon>0$, $Z + i \epsilon$ still lies in the upper half-plane for any $Z$. On the other hand, negative values of $\epsilon$  may map some states into the \emph{lower} half-plane. As such, the corresponding transformation is not physical. But \(S_{-\epsilon}\) is the inverse of \(S_{\epsilon}\), hence the latter only belongs to a semigroup.
For example, in the one-mode case, if \(Z=i/2\) and \(\epsilon=-1\) then \(\phi_{S_{\epsilon}}(Z)=Z+i\epsilon=-i/2\), which lies outside the upper half-plane.
Interpreting \(M_{\epsilon}\propto e^{-\epsilon \q^{2}}\) as a Gaussian Kraus operator, it defines a single-Kraus Gaussian operation \(\rho\mapsto M_{\epsilon}\rho M_{\epsilon}^{\dag}\), which is generally trace-nonincreasing (non-deterministic). The explicit prefactor \(n(\epsilon)\) corresponds to conditioning on this Kraus outcome and renormalizing the resulting state; at the level of density matrices this conditioning is nonlinear, but it induces a well-defined update on the (pure-state) half-plane representative \(Z\) that we track here. In these conventions the physical constraint is \(\epsilon\ge0\), ensuring \(Z\mapsto Z+i\epsilon\) preserves \(\operatorname{Im}Z>0\).

This motivates restricting the set of complex symplectic matrices to those that preserve the upper half-plane. The appropriate restriction is a Siegel-domain preserving semigroup \cite{hilgert_note_1989,derezinski_quantization_2020}.

\begin{definition}[Siegel-domain preserving semigroup (upper half-plane)]\label{def:contraction-semigroup-uhp}
	Define
	\begin{align}\label{eq:sympplusdef}
		\symplusn \coloneqq \{T \in \symplecticcomplexn\ |\ T^{\dag} \sigma_{y} T \ge \sigma_{y} \}
		\,,
	\end{align}
	where the inequality is a positive-semidefinite (matrix) inequality (recall that \(\sigma_{y}\) is Hermitian by convention).
\end{definition}
Although just write \(\symplusn\) omitting the field, it is helpful to remember that it is defined as a subsemigroup of the \emph{complex} symplectic group \(\symplecticcomplexn\) via \cref{eq:sympplusdef}. A simple figure to help keep track of notation and set inclusions at a glance is \cref{fig:symplusn-inclusion}.

The direction of the inequality is a convention: replacing \(\sigma_{y}\) by \(-\sigma_{y}\) swaps \(\ge\) and \(\le\). What matters is that \(\symplusn\) is exactly the set of complex symplectic matrices whose associated linear fractional transformation \(\phi_{T}\) preserves the Siegel upper half-plane.
\begin{center}
	\begin{minipage}{\linewidth}
	\centering
	\begin{tikzpicture}[every node/.style={font=\scriptsize,align=center},>=stealth]
		\node[draw,rounded corners,minimum width=5.4cm,minimum height=2.9cm,inner sep=6pt] (outer) {};
		\node[anchor=north] at (outer.north) {\(\symplecticcomplexn\)};

		\node[draw,rounded corners,fill=gray!12,very thick,minimum width=4.2cm,minimum height=1.85cm,inner sep=5pt] (inner) at (outer.center) {\(\symplusn\)};
		\draw[->] (inner.south) -- ++(0,-0.65) node[below] {\(\partial\symplusn=\symplecticn\)};
	\end{tikzpicture}
		\captionof{figure}{Schematic inclusion \(\symplecticn\subset\symplusn\subset\symplecticcomplexn\). The complex symplectic group \(\symplecticcomplexn\) consists of \(2n\times2n\) complex matrices satisfying \(S^{T}\sigma_{y}S=\sigma_{y}\). The Siegel-domain preserving semigroup is \(\symplusn=\{T\in\symplecticcomplexn\mid T^{\dag}\sigma_{y}T\ge\sigma_{y}\}\); its induced fractional transformations preserve the Siegel upper half-plane \(\siegeln=\{Z=Z^{T},\ \operatorname{Im}Z>0\}\). The boundary \(\partial\symplusn=\symplecticn\) (where the inequality saturates) recovers Gaussian unitaries.}
	\label{fig:symplusn-inclusion}
	\end{minipage}
\end{center}

\begin{theorem}[Siegel-domain preserving semigroup preserves the Siegel upper half-plane]\label{thm:sympplus-preserves-uhp}
	If \(T\in\symplusn\), then \(\phi_{T}(\siegeln)\subseteq\siegeln\).
\end{theorem}

This is a special case of the general Lie-semigroup theory of Hermitian symmetric domains: the (Olshanskii) compression semigroup in the complexified group \(G_{\complex}\) can be characterized as the set of elements mapping the domain into itself. To the best of our knowledge, the crispest single statement of this identification is Koufany's Theorem~6.1~\cite{koufany_jordan_2006}, attributed there to Olshanskii~\cite{olshanskii_complex_1991}; specializing to the Siegel domain \(\siegeln\) and \(G_{\complex}=\symplecticcomplexn\) yields the claim.\footnote{In the oscillator-semigroup literature, the same domain-preserving property is often treated as folklore and is implicit in the construction (see, e.g., Refs.~\cite{howe_oscillator_1988,hilgert_note_1989}). For example, Folland's development of the (normalized) oscillator semigroup is organized around Gauss kernels parametrized by a Siegel domain, and would be difficult to reconcile with the Gaussian-state update picture if the induced holomorphic action failed to preserve the domain; however, we have not found a single isolated lemma/proposition there stating \(\phi_T(\siegeln)\subseteq\siegeln\) explicitly in semigroup form.}
The boundary case \(\partial\symplusn=\symplecticn\) recovers Gaussian unitaries and is classical~\cite{folland_harmonic_1989}.

These operations are realized on Hilbert space by Gaussian Kraus operators in the \emph{oscillator semigroup}. Concretely, in the position basis each element is a Gauss-kernel operator, whose integral kernel is determined by a symmetric matrix \(\mathcal Z\). Since the kernel depends on both input and output variables \((x,y)\in\reals^{2n}\), the parameter \(\mathcal Z\) naturally lives in the Siegel domain of doubled dimension (i.e., the same definition as \(\Sigma_{n}\) but with \(n\mapsto 2n\)), which we denote by \(\siegelN\).

\begin{definition}[Normalized oscillator semigroup]\label{def:normalized-oscillator-semigroup}
	Following~\cite{folland_harmonic_1989}, the \emph{normalized oscillator semigroup} is the set of Gauss-kernel operators
	\begin{align}\label{eq:oscsemigrouphalfplane}
		\Omega^{0} = \Bqty{\hat V_{\mathcal Z}\ |\ \mathcal Z \in \siegelN}
		\,.
	\end{align}
	Here \( \hat V_{\mathcal Z} \) is a Gauss-kernel operator, whose matrix element is
	\begin{align}\label{eq:gaussmateluhp}
		\bra{x} \hat V_{\mathcal Z} \ket{y} =
		\exp(\frac{i}{2} (x^{T}, y^{T}) \mathcal Z (x^{T}, y^{T})^{T})
		\,,
	\end{align}
	and we fix a canonical overall scalar normalization.
\end{definition}
Elements of \(\Omega^{0}\) are Gaussian Kraus operators, hence they implement single-Kraus Gaussian operations on pure states (with conditional renormalization). This goes beyond the unitary metaplectic/symplectic action (recovered on the boundary), but it is still not large enough to represent general deterministic Gaussian channels, which involve sums (or integrals) over Kraus operators.
Equation \cref{eq:gaussmateluhp} simply fixes conventions in a concrete basis; later we will mostly use the same semigroup \(\Omega^{0}\) in the Fock--Bargmann representation, where Hilgert's parametrization makes the relation to disk fractional transformations transparent~\cite{hilgert_note_1989}. The normalization in \cref{eq:oscsemigrouphalfplane} fixes an overall scalar so that composition is well behaved~\cite{folland_harmonic_1989,derezinski_quantization_2020}; it is unrelated to trace normalization (in particular, typically \(\tr(\hat V_{\mathcal Z})\neq 1\)).
In the terminology above, these are non-deterministic Gaussian operations: they are naturally described by Kraus operators (or instruments), and their ``renormalized'' action on normalized states corresponds to conditioning on an outcome.

With this normalization fixed, \(\Omega^{0}\) projects onto \(\symplusn\) with the usual \(\pm1\) ambiguity, in direct analogy with the metaplectic/symplectic double cover (see, e.g., Refs.~\cite{howe_oscillator_1988,hilgert_note_1989}). In this way, \(\symplusn\) provides the fractional action on \(\siegeln\), while \(\Omega^{0}\) provides its Hilbert-space realization; \cref{fig:oscsemigroup-schematic} summarizes the relationship.
\begin{center}
	\begin{minipage}{\linewidth}
	\centering
		\begin{tikzpicture}[>=stealth,every node/.style={font=\scriptsize,align=center},
			row sep=1.6cm, column sep=1.4cm]
				\node (siegel2) at (0,0) {\(\siegelN\)};
				\node (omega) at (2,0) {\(\Omega^{0}\)};
				\node (symplus) at (4,0) {\(\symplusn\)};
				\node (siegel) at (6.5,0) {\(\siegeln\cong\symplecticn/K(n)\)};
			\node (mp) at (2,-1.6) {\(\mathrm{Mp}(2n,\reals)\)};
			\node (sp) at (4,-1.6) {\(\symplecticn\)};
				\draw[<->] (siegel2) -- node[above]{\(\cong\)} (omega);
				\draw[->>] (omega) -- node[above]{\(\pm\)} (symplus);
				\draw[->] (symplus) -- node[above]{\(\curvearrowright\)} (siegel);
				\draw[->>] (mp) -- node[above]{\(\pm\)} (sp);
				\draw[->] (sp) -- node[above right,inner sep=1pt]{\(\curvearrowright\)} (siegel);
				\draw[->,dashed] (symplus.south) -- node[left]{\(\partial\)} (sp.north);
		\end{tikzpicture}
				\captionof{figure}{Oscillator-semigroup dictionary (upper half-plane). Elements of the normalized oscillator semigroup \(\Omega^{0}\) are Gaussian Kraus operators (Gauss-kernel operators) parametrized by the doubled Siegel domain \(\siegelN=\Sigma_{2n}\). Modulo the standard \(\pm\) ambiguity, \(\Omega^{0}\) projects to the matrix semigroup \(\symplusn\subset\symplecticcomplexn\). The notation \(G\curvearrowright X\) denotes an action: each \(T\in G\) determines an endomorphism \(\phi_T\colon X\to X\), and evaluating it gives the update \(Z\mapsto \phi_T(Z)\). In our case \(X=\siegeln\) and \(\phi_T\) is the linear fractional map induced by \(T\). This is the pure-state precursor of the channel-level picture developed later.}
	\label{fig:oscsemigroup-schematic}
	\end{minipage}
	\end{center}
	Our main result generalizes this picture from single-Kraus (non-deterministic) Gaussian operations on pure states to \emph{deterministic} Gaussian channels acting on density matrices. In the disk formulation, points of \(\siegeldiskN\) (the disk counterpart of \(\siegelN\)) naturally parametrize oscillator-semigroup elements; a key step is to identify the subset corresponding to physical density operators.

\subsection{The Siegel Disk}\label{sec:siegel-disk}
The central object of interest in this article is the adjacency matrix associated to states in Fock-Bargmann representation. This matrix lives in the Siegel disk, $\siegeldiskn$, defined in \cref{eq:gaussian_fock-bargmann}. Each half-plane representative of a pure gaussian state corresponds to a disk representative. They are related by the Cayley transform,\footnote{Strictly speaking there is no such thing as \emph{the} Cayley transform: many linear fractional coordinate changes go by that name. In the Siegel-space literature, however, there is a standard ``relevant'' one intertwining the upper half-plane and disk pictures, and we will follow this convention and refer to it as \emph{the} Cayley transform when no ambiguity can arise. Since the fixed matrix \(\Gamma\) is the one that implements this coordinate change (by conjugation), we refer to it as the \emph{Cayley matrix}---Cayley as in ``Cayley transform'', not as an attribution that Arthur Cayley would necessarily insist on.} a linear fractional transformation with respect to the Cayley matrix,
\begin{align}
	K
	=
	\phi_{\Gamma}(Z)
	\label{eq:kfromz}
	\,,
\end{align}
which is the effect of the Bargmann transform on the adjacency matrix of a state in position representation. Here \(\phi_{\Gamma}\) is a change of coordinates between the half-plane and disk representatives (not a dynamical map). An adjacency matrix $K \in \siegeldiskn$ belongs to the disk if and only $Z = \phi^{-1}_{\Gamma}(K)$ belongs to the  half-plane, $\siegeln$. We can find the transformation law for $K$ with respect to a gaussian unitary associated to $S_{r} \in \symplecticn$ by inspecting the Cayley transform of $Z \in \siegeln$ under the action of $S_{r}$. Using the composition law, \cref{eq:compositionlaw}, we have
$
	\phi_{\Gamma}(\phi_{S_{r}}(Z))
	=
	\phi_{\Gamma S_{r} \Gamma^{\dag}} \pqty{K} $.
Thus, if \(Z \in \siegeln\) transforms under the fractional transformation induced by \(S_{r} \in \symplecticn \), $K \in \siegeldiskn$ transforms under the fractional transformation induced by its ABC version $S = \Gamma S_{r} \Gamma^{\dag}$.

\begin{center}
	\begin{minipage}{\linewidth}
	\centering
	\begin{tikzpicture}[x=1.3cm,y=1.3cm,>=stealth,every node/.style={font=\scriptsize}]
		\node[font=\scriptsize] at (0,1.35) {\(\Delta_{1}\)};
		\fill[gray!8] (0,0) circle (1);
		\draw[thick,dashed] (0,0) circle (1);
		\fill (0,0) circle (1.2pt) node[below left] {\(K_{0}=0\)};
		\fill (0.65,0.3) circle (1.2pt) node[above left] {\(K_{1}\)};
		\draw[thick,->] (0,0) -- (0.65,0.3) node[midway,above] {\(S\)};
		\draw[thin,|-|] (0,-1.15) -- node[below]{\(1\)} (1,-1.15);
	\end{tikzpicture}
		\captionof{figure}{Disk intuition (\(n=1\)). The Siegel disk is \(\siegeldiskn=\{K=K^{T},\ I-K^{\dag}K>0\}\); for one mode it reduces to the unit disk \(\{k\in\complex\mid |k|<1\}\). The vacuum corresponds to \(K_{0}=0\). A Gaussian unitary \(S\in\symcayn\) acts by a M\"obius-type fractional transformation \(K\mapsto\phi_{S}(K)\), which preserves \(\siegeldiskn\) and is related to the upper half-plane action by the Cayley transform.}
	\label{fig:disk-intuition}
	\end{minipage}
\end{center}
	For commutative-diagram presentations of the Cayley transform intertwining the (unitary) symplectic action on the upper half-plane and disk, see, e.g., Refs.~\cite{menicucci_graphical_2011,gabay_passive_2016,freitas_revisiting_2004,koufany_jordan_2006}.

In terms of the ABC matrices $S$, fractional transformations in the disk are
\begin{align}
	\label{eq:linearfractional_K}
	K' & = \phi_{S}(K)
	\nonumber                                               \\
	S  & \in \symcayn = \Gamma (\symplecticn) \Gamma^{\dag}
	\,,
\end{align}
where we introduced the ABC symplectic group \(\symcayn\). A presentation for this group is
\begin{align}
			\label{eq:defsympcomp}
			S \sigma_{z} S^\dag & = \sigma_{z}
		\\
			\label{eq:defblockym}
			S \sigma_{x}        & = \sigma_{x} S^{*}
			\,,
	\end{align}
	where \(\sigma_{x}\) and \(\sigma_{z}\) are the first and third Pauli matrices.
\begin{definition}[ABC symplectic group]\label{def:symcayn}
	The ABC symplectic group \(\symcayn\) is the Cayley-conjugate of the real symplectic group:
	\begin{align}
		\symcayn \coloneqq \Gamma \symplecticn \Gamma^{\dag}
		\,.
	\end{align}
	Equivalently, it is the set of all \(2n\times 2n\) complex matrices \(S\) satisfying \cref{eq:defsympcomp,eq:defblockym}. 
\end{definition}
Despite consisting of complex matrices, \(\symcayn\) is just a change-of-coordinates copy of the \emph{real} symplectic group \(\symplecticn=\mathrm{Sp}_{2n}(\reals)\). This is why we use the superscript \(\Gamma\) rather than writing it as another ``\(\mathrm{Sp}_{2n}(\complex)\)'' object, and it helps distinguish \(\symcayn\) from the genuinely complex semigroup \(\symplusn\subseteq\mathrm{Sp}_{2n}(\complex)\) introduced in \cref{def:contraction-semigroup-uhp}.
As in the upper half-plane picture, it is often practical to represent Gaussian transformations in block form. Any ABC real matrix is in the form
\begin{align}
	S
	 &
	=
	\begin{pmatrix}
		\alpha  & \beta    \\
		\beta^* & \alpha^* \\
	\end{pmatrix}
	\,.
\end{align}
[This last statement is equivalent to~\cref{eq:defblockym}.]
In block form, the presentation of \(\symcayn\) reads
\begin{align}
	\label{eq:alphabetaone}
	\alpha ^\dag \alpha - \beta ^T \beta ^*
	 & =
	1
	\\
	\label{eq:alphabetatwo}
	\alpha ^T \beta ^*
	 & =
	\beta ^\dag \alpha
	\,,
\end{align}
or, alternatively
\begin{align}
	\label{eq:alphabetaonealt}
	\alpha \alpha ^\dag - \beta \beta ^\dag
	 & =
	1
	\\
	\label{eq:alphabetatwoalt}
	\alpha \beta^T
	 & =
	\beta \alpha^T
	\,.
\end{align}
We can write the formula in \cref{eq:linearfractional_K} explicitly as
\begin{align}
	\label{eq:linearfractional}
	\phi_{S}(K)
	=
	( \alpha^{*} K + \beta^{*} )
	( \beta K  + \alpha )^{-1}
	\,,
	\end{align}
	where $\alpha$ and $\beta$ are a symplectic pair of blocks that satisfy \cref{eq:alphabetaone,eq:alphabetatwo,eq:alphabetaonealt,eq:alphabetatwoalt}. From now on the symbols $\alpha$ and $\beta$ are always understood as a symplectic pair in this sense. Also note that the fractional transformation thus defined still follows the ``DCBA'' rule.

	In perfect analogy with the Siegel upper half-plane picture, the (unitary) ABC symplectic action preserves the disk and acts transitively on it. We sketch the salient parts of the argument, as the stacked-notation manipulations will be useful later on. For completeness, we state the formal result later as \cref{thm:symcayn-preserves-disk-transitive} and give a complete proof in Appendix \cref{app:thm-symcayn-preserves-disk-transitive}. Using stacked notation, the conditions required by $K$ to belong to the disk are as follows:
\begin{align}
	\label{eq:siegelreq1}
	\begin{bmatrix}
		1 \\
		K \\
	\end{bmatrix}^{T}
	\sigma_{y}
	\begin{bmatrix}
		1 \\
		K \\
	\end{bmatrix}
	 & =
	0
	\\
	\label{eq:siegelreq2}
	\begin{bmatrix}
		1 \\
		K \\
	\end{bmatrix}^{\dag}
	\sigma_{z}
	\begin{bmatrix}
		1 \\
		K \\
	\end{bmatrix}
	 & > 0
	\,.
\end{align}
To prove that any \(S\in\symcayn\) maps the Siegel disk into itself, we show that the stacked representative \(S\begin{bmatrix}1\\ K\end{bmatrix}\) determines a matrix \(K'=\phi_S(K)\) that again satisfies \cref{eq:siegelreq1,eq:siegelreq2}. Indeed, we find
\begin{align}
	\label{eq:siegelreqproof1}
	\begin{bmatrix}
		1 \\
		K \\
	\end{bmatrix}^{T}
	S ^{T}
	\sigma_{y}
	S
	\begin{bmatrix}
		1 \\
		K \\
	\end{bmatrix}
	 & =
	0
	\\
	\label{eq:siegelreqproof2}
	\begin{bmatrix}
		1 \\
		K \\
	\end{bmatrix}^{\dag}
	S ^{\dag}
	\sigma_{z}
	S
	\begin{bmatrix}
		1 \\
		K \\
	\end{bmatrix}
	 & > 0
	\,.
\end{align}
The inequality is immediately satisfied as
$
	S ^{\dag}
	\sigma_{z}
	S
	= \sigma_{z}
$,
and \cref{eq:siegelreq2} holds for \(K'\). The equality is also satisfied since
$
	S ^{T} \sigma_{y} S
	=
	i (S ^{\dag} \sigma_{z} S)^{*} \sigma_{x}
	=
	\sigma_{y}
$, and \cref{eq:siegelreq1} holds for \(K'\).
This essentially proves that fractional transformations map the disk into itself (we have omitted some technicalities that we will address later).

Transitivity on the Siegel disk means that any \(K_{1}\) and \(K_{2}\in\siegeldiskn\) differ by a fractional transformation with respect to \(\symcayn\). In particular, for any \(K\in\siegeldiskn\) there exists \(S\in\symcayn\) such that \(\phi_S(0)=K\); an explicit choice is given in the proof of \cref{thm:symcayn-preserves-disk-transitive} in Appendix \cref{app:thm-symcayn-preserves-disk-transitive}.

A state $\ket{K}$ with Fock-Bargmann wavefunction as in~\cref{eq:gaussian_fock-bargmann} is Gaussian if and only if $K \in \siegeldiskn$.
The action of Gaussian unitaries on the corresponding Fock-Bargmann wavefunction is given by the Fock-metaplectic representation. The Siegel disk, like the upper half-plane, allows us to avoid the explicit representation of the Fock-metaplectic group and to think of Gaussian operations as fractional transformations on a symmetric space instead.
\begin{definition}[Fock-metaplectic representation]\label{def:fock-metaplectic}
	The (projective) Fock-metaplectic representation assigns to each \(S\in\symcayn\) a unitary \(\op U_{S}\) (unique up to an overall phase) acting on the Fock--Bargmann Hilbert space such that, on pure Gaussian states,
	\begin{align}\label{eq:actiononstate}
		\op U_{S} \ket {K} \propto \ket{{\phi_{S}(K)}}
		\,,
	\end{align}
		where the normalization is unnecessary since we will only track the induced action \(\phi_{S}\). See, e.g., Refs.~\cite{howe_oscillator_1988,hilgert_note_1989,simon_gaussian_1988}.
\end{definition}

Among the transformations in the Siegel disk we can find the maximal compact subgroup of the symplectic group, corresponding to phase-space rotations. When $\beta = 0$, \cref{eq:alphabetaone} gives $\alpha \alpha^\dag = 1$, so that the set of these transformations is $U(n) \cong K(n)$---the well-known subgroup~\cite{simon_quantumnoise_1994} of the symplectic group that leaves the vacuum invariant. The corresponding fractional transformation is
$
	\phi_{\alpha} (K)
	=
	\alpha
	K
	\alpha^{T}
$.

We move on to the non-unitary evolution of pure states. As in the half-plane case, the disk picture admits a class of single-Kraus Gaussian operations described by fractional transformations \(\phi_{T}\), where \(T\) belongs to a currently unspecified matrix set. The quickest way to see this, as well as identify the set, is to evolve an adjacency matrix \(K\) by \(T\) and require that \(K'=\phi_{T}(K)\) remain in \(\siegeldiskn\): it suffices to re-check \cref{eq:siegelreqproof1,eq:siegelreqproof2}. Both the equality and the inequality can be satisfied if \(T^{T} \sigma_{y} T = \sigma_{y}\) and \(T^{\dag} \sigma_{z} T \ge \sigma_{z}\). This motivates the following definition.
		\begin{definition}[Siegel-disk preserving semigroup]\label{def:sympluscayn}
			Define
			\begin{align}\label{eq:sympcomplexdef}
			\sympluscayn \coloneqq
			\{ T \in \symplecticcomplexn \ \big|\ T^{\dag} \sigma_{z} T \ge \sigma_{z} \}
			\,.
		\end{align}
					It induces fractional transformations mapping \(\siegeldiskn\) into itself, with boundary \(\symcayn\). See, e.g., Refs.~\cite{hilgert_note_1989,derezinski_quantization_2020}.
			\end{definition}
				\begin{center}
					\begin{minipage}{\linewidth}
					\centering
				\begin{tikzpicture}[every node/.style={font=\scriptsize,align=center}]
					\node[draw,rounded corners,minimum width=5.4cm,minimum height=2.9cm,inner sep=6pt] (outer) {};
					\node[anchor=north] at (outer.north) {\(\symplecticcomplexn\)};

					\node[draw,rounded corners,fill=gray!12,very thick,minimum width=4.2cm,minimum height=1.85cm,inner sep=5pt] (inner) at (outer.center) {\(\sympluscayn\)};
					\draw[->,>=stealth] (inner.south) -- ++(0,-0.65) node[below] {\(\partial\sympluscayn=\symcayn\)};
				\end{tikzpicture}
					\captionof{figure}{Disk-preserving semigroup. The Siegel-disk preserving semigroup is \(\sympluscayn=\{T\in\symplecticcomplexn\mid T^{\dag}\sigma_{z}T\ge\sigma_{z}\}\) (with \(T^{T}\sigma_{y}T=\sigma_{y}\) automatically since \(\sympluscayn\subset\symplecticcomplexn\)). Its induced fractional transformations map the Siegel disk \(\siegeldiskn=\{K=K^{T},\ I-K^{\dag}K>0\}\) into itself, describing single-Kraus Gaussian operations on pure states in disk coordinates. The boundary \(\partial\sympluscayn=\symcayn\) (equality \(S^{\dag}\sigma_{z}S=\sigma_{z}\)) recovers Gaussian unitaries in the ABC/disk realization.}
					\label{fig:sympluscayn-inclusion}
					\end{minipage}
				\end{center}
	
\begin{theorem}[Siegel-disk preserving semigroup preserves the Siegel disk]\label{thm:sympplus-preserves-disk}
			If \(T\in\sympluscayn\), then \(\phi_{T}(\siegeldiskn)\subseteq\siegeldiskn\).
\end{theorem}
\begin{proof}
	This proof uses stacked notation with a useful subtlety: symmetry uses the \emph{transpose} symplectic condition \(T^{T}\sigma_{y}T=\sigma_{y}\), while disk membership uses the \emph{adjoint} inequality \(T^{\dag}\sigma_{z}T\ge\sigma_{z}\) (and \(\sigma_{y}^{T}=-\sigma_{y}\) whereas \(\sigma_{z}^{\dag}=\sigma_{z}\); see \cref{def:lifted-paulis}).

		Let \(K\in\siegeldiskn\) and write \(T=\begin{psmallmatrix}A&B\\ C&D\end{psmallmatrix}\). Symmetry of \(K'=\phi_T(K)\) follows from \(T\in\symplecticcomplexn\) as in \cref{thm:complex-symplectic-preserves-symmetry}. For the disk inequality, set \(X\coloneqq\begin{bmatrix}1\\ K\end{bmatrix}\), so that \(X^{\dag}\sigma_z X = 1-K^{\dag}K>0\). Since \(T^{\dag}\sigma_z T\ge\sigma_z\), we have \((TX)^{\dag}\sigma_z(TX)\ge X^{\dag}\sigma_z X>0\). Writing \(TX=X'(BK+A)\) with \(X'=\begin{bmatrix}1\\ K'\end{bmatrix}\) (stacked action), this implies \(BK+A\) is invertible and \(X'^{\dag}\sigma_z X'>0\), i.e. \(1-K'^{\dag}K'>0\). A complete stacked-notation proof (including invertibility) is given in \cref{app:thm-sympplus-preserves-disk}.
\end{proof}
We therefore view \(\sympluscayn\) as the semigroup of fractional transformations that maps \(\siegeldiskn\) into itself. The group \(\symcayn\) is the boundary of this semigroup. This is slightly less immediate than the corresponding statement \(\partial\symplusn=\symplecticn\), so we sketch a short argument. At the boundary, the inequality saturates, i.e. \(S^{\dag}\sigma_{z}S=\sigma_{z}\), hence
\[
	S^{-1}=\sigma_{z}S^{\dag}\sigma_{z}
	\,.
\]
Combining this with the complex symplectic condition \(S^{T}\sigma_{y}S=\sigma_{y}\) and \(\sigma_{y}=i\sigma_{x}\sigma_{z}\) gives \(S^{T}\sigma_{x}=\sigma_{x}S^{\dag}\), and taking adjoints yields \(S\sigma_{x}=\sigma_{x}S^{*}\), i.e. \(S\in\symcayn\). Conversely, any \(S\in\symplecticcomplexn\) satisfying \(S\sigma_{x}=\sigma_{x}S^{*}\) obeys \(S^{\dag}\sigma_{z}S=\sigma_{z}\), so \(\symcayn\) is exactly the boundary of \(\sympluscayn\).

For completeness, we also state the classic unitary (boundary) statement in the disk picture:
\begin{theorem}[ABC symplectic action preserves the Siegel disk and is transitive]\label{thm:symcayn-preserves-disk-transitive}
	If \(S\in\symcayn\), then \(\phi_{S}(\siegeldiskn)\subseteq\siegeldiskn\). Moreover, for any \(K_{1},K_{2}\in\siegeldiskn\), there exists \(S\in\symcayn\) such that \(K_{2}=\phi_{S}(K_{1})\) (transitivity).
\end{theorem}
The main algebraic identities are sketched earlier around \cref{eq:siegelreqproof1,eq:siegelreqproof2}; we tie up the remaining details in Appendix \cref{app:thm-symcayn-preserves-disk-transitive}.

Thus, in the disk picture, fractional transformations in general realize single-Kraus Gaussian operations on pure states
\begin{align}\label{eq:nonunitaryevolution}
	K' & = \phi_{T} (K)
	\nonumber           \\
	T  & \in \sympluscayn
	\,.
\end{align}

Finally, we discuss the oscillator semigroup in the Fock--Bargmann / Siegel disk picture. This is not a new object: it is the same normalized oscillator semigroup \(\Omega^{0}\) from \cref{def:normalized-oscillator-semigroup}, now written in a different basis (and following Hilgert's explicit parametrization~\cite{hilgert_note_1989})---in exactly the same sense that the Siegel disk is not a different symmetric space from the upper half-plane, but merely the same one in different coordinates. In the Fock--Bargmann representation, a Gaussian Kraus operator $\op K \in \mc{B}(\fockbargn)$ is a bounded linear operator acting on states as
\begin{align}
	\label{eq:Ttensor}
	\bra{\zeta} \op K \ket{\psi} = \int d\mu(\omega) K(\zeta,\omega^{*}) \psi(\omega)
	\,,
\end{align}
where $K(\zeta,\omega^{*}) = \bra{\zeta} \op K \ket{\omega}$ is the FB matrix element of $\op K$ and $\psi(\omega) = \bra{\omega} \ket{\psi}$ is the wavefunction in FB representation. We can parametrize the set of \emph{Gauss-kernel} operators by a symmetric matrix
\begin{align}\label{eq:xmatrix}
	\amat =
	\begin{pmatrix}
		A     & B \\
		B^{T} & D
	\end{pmatrix}
	\,,
\end{align}
where $A$ and $D$ are complex symmetric matrices, and each submatrix in the above is $n \times n$. For each $\amat$ there is a family of Gauss-kernel operators whose matrix element is
\begin{align}\label{eq:gausskernels}
	K_{\amat}(\zeta, \omega^{*})
	 & = c \exp(
	\half \pqty{\zeta^{T} A \zeta + 2 \zeta^{T} B \omega^{*} + \omega^{\dag} D \omega^{*}}
	)
	\\ &=
	c \exp(\half
	\begin{pmatrix}
			\zeta^{T} & \omega^{\dag}
		\end{pmatrix}
	\amat
	\begin{pmatrix}
			\zeta \\ \omega^{*}
		\end{pmatrix}
	)
	\,,
\end{align}
with \(c\) a nonzero complex number. The set of linear operators induced by kernels in this form is the oscillator semigroup in the disk~\cite{hilgert_note_1989}. The crucial observation in connecting this semigroup and fractional transformations is that, after fixing the overall scalar normalization (normalized oscillator semigroup), it becomes a double-cover of \(\sympluscayn\), and the semigroup homomorphism between the latter and the former is~\cite{hilgert_note_1989}
\begin{definition}[Oscillator semigroup (disk picture)]\label{def:oscillator-semigroup-disk}
	The oscillator semigroup in the disk picture is the semigroup of linear operators induced by Gauss kernels of the form \cref{eq:gausskernels}. Its normalized version is (projectively) related to the disk-preserving semigroup \(\sympluscayn\) via the homomorphism \cref{eq:homosym}; see, e.g., Refs.~\cite{hilgert_note_1989,howe_oscillator_1988}.
\end{definition}
\begin{align}
	\label{eq:homosym}
	T(\amat) =
	\begin{pmatrix}
		B^{-T}     & B^{-T} D       \\
		- A B^{-T} & B - A B^{-T} D \\
	\end{pmatrix}
	=
	\begin{pmatrix}
		\alpha & \beta  \\
		\gamma & \delta
	\end{pmatrix}
	\,,
\end{align}
where $T \in \sympluscayn$ defines the fractional transformations \(\phi_{T}\) acting on \( \siegeldiskn \).
	The minus sign relative to Ref.~\cite{hilgert_note_1989} is due to the fact that they define Gaussian kernels with a minus sign in the exponent. The inverse, \( \amat (T) \), always exists, and it is
\begin{align}
	\label{eq:homosyminv}
	\amat(T) =
	\begin{pmatrix}
		\gamma \alpha^{-1} & \alpha^{-T}         \\
		\alpha^{-1}        & - \alpha^{-1} \beta \\
	\end{pmatrix}
	\,.
\end{align}
Given the kernel \((\det B)^{-1/2} K_{\amat}(\zeta,\omega^{*})\), where \(\amat \in \siegeldiskN\), and a pure Gaussian state \cref{eq:gaussian_fock-bargmann}, one can verify by Gaussian integration that the induced fractional transformation on Gaussian states is indeed \(\phi_{T(\amat)}\). Because of the two-to-one correspondence between \(\siegeldiskN\) and \(\symcayn\), we may think of fractional transformations on pure Gaussian states as elements of a \(2n\)-dimensional disk (operations) acting on an \(n\)-dimensional one (states).

At the state level, $\op K_{\amat}$ is a representation of the oscillator semigroup with a simple action on Gaussian states:
\begin{align}\label{eq:actionongaussianosc}
	\op K_{\amat(T)} \ket{K} \propto \ket{\phi_{T }(K)}
	\,.
\end{align}
We recover Gaussian unitaries in the disk by letting \(S\) lie on the boundary, \(S\in\symcayn\). In that case, the assignment \(S\mapsto \op K_{\amat(S)}\) (equivalently \(S\mapsto \op U_{S}\)) is a projective representation of \(\symcayn\), i.e., the Fock-metaplectic representation, whose transformation law on pure Gaussian states is \cref{eq:actiononstate}. More precisely, the normalized oscillator semigroup is a double cover of \(\sympluscayn\), so choosing an operator over a given \(S\in\sympluscayn\) is only defined up to an overall sign; restricting to the boundary recovers the familiar \(\pm1\) ambiguity of the metaplectic lift of the symplectic group. This subtlety bears no consequences for our purposes.

\section{Identifying the double disk representative}
We have seen that pure states admit a Siegel-space representation, be it the disk or the plane. The plane description is built from the wavefunction of Gaussian pure states, while the disk description uses their Fock-Bargmann wavefunction.

The oscillator semigroup illustrates how a subset of quantum operations---operators (or kernels) with both input and output variables, rather than state vectors---can be associated to a larger disk. In what follows, we extend this picture to accommodate a more general situation where gaussianity-preserving maps act on mixed Gaussian states.
We can carry out this program by finding mixed Gaussian states (resp.\ Gaussian operations/channels) as a subset of a \(2n\)-dimensional (resp.\ \(4n\)-dimensional) disk.

For pure states, \( K \in \siegeldiskn \) was the adjacency matrix of the Gaussian Fock-Bargmann wavefunction. One generalization of this is to consider the FB \emph{kernel} of a mixed state, which is a Gaussian function~\cite{hilgert_note_1989} with an adjacency matrix \( \amat \in \siegeldiskN \)
\begin{align}\label{eq:defstellar}
	\bra{ \zeta } \rho_{\amat} \ket{ \zeta }
	=
	\exp(\half
	\begin{pmatrix}
			\zeta^{T} & \zeta^{\dag}
		\end{pmatrix}
	\amat
	\begin{pmatrix}
			\zeta \\ \zeta^{*}
		\end{pmatrix})
	\,.
\end{align}
All we need to do now is to restrict \(\amat\) to a subset of \(\siegeldiskN\) for which \(\rho_{\amat}\) is a bona fide density matrix (Hermitian, positive, and trace-normalized).
The condition \(\amat\in\siegeldiskN\) ensures that the corresponding Gaussian-kernel operator is trace-class. While one can enforce trace normalization and positivity directly at the kernel level (via Gaussian integrals), we will instead impose these conditions using the covariance-matrix characterization and transport them to \(\amat\) using the fractional coordinate change \(\sigma\leftrightarrow \amat\) introduced below. This lets us sidestep explicit integration, one of the main benefits of the framework.
For completeness: trace-class/boundedness and trace formulas for Gaussian kernels are standard consequences of Gaussian integration (see, e.g., Ref.~\cite{hilgert_note_1989} and the general ``quantization of Gaussians'' trace/norm formulas in Ref.~\cite{derezinski_quantization_2020}). Moreover, Hermiticity \(\rho_{\amat}^{\dag}=\rho_{\amat}\) is equivalent to the kernel symmetry \(K_{\amat}(\zeta,\omega^{*})=\overline{K_{\amat}(\omega,\zeta^{*})}\), which in matrix form becomes \(\amat\sigma_{x}=\sigma_{x}\amat^{*}\).

Instead, here we relate the definition of \( \amat \) to others given in recent work \cite{yao_recursive_2022,miatto_fast_2020,kruse_detailed_2019}. In these works, the authors define a certain matrix (essentially, \( \amat \)) and build procedures based on its manipulation that are useful to solve problems in continuous-variable quantum information. Each of these works conveniently relate their definitions to the Husimi function. With the shorthand \( z = (\zeta^{T}, \zeta^{\dag})^{T}\), the Husimi function for a gaussian state is
\begin{align}\label{eq:husimigauss}
		Q_{\sigma_{Q}}(z) = \pi^{-n}\det\pqty{\sigma_{Q}}^{-1/2}\exp\pqty{ -\frac{1}{2} z^{\dag} \sigma_{Q}^{-1} z }
		\,,
	\end{align}
and we recall the relationship between the FB kernel and the Husimi function of a state,
\begin{align}
	K (\zeta,\zeta^{*}) &= \pi^{-n} e^{\abs{\zeta}^{2}} Q(\zeta^{*})
	\\
	&= \pi^{-n} \exp[
		\frac{1}{2} \abs{ z }^{2}
		-\frac{1}{2} z^{T} \sigma_{Q}^{-1} z^{*}]
	\,.
\end{align}
With the same shorthand for \(z\), the FB kernel in \cref{eq:defstellar} is
\(  \exp(\frac{1}{2} z^{T} \amat z) \), implying
\(  \amat \sigma_{x} = 1 - \sigma_{Q}^{-1} \), or
\begin{align}\label{eq:sigmaqamat}
	\amat^{*} = \sigma_{x} \pqty{ 1 - \sigma_{Q}^{-1} }
	\,,
\end{align}
which is the the complex conjugate of the definitions found in refs. \cite{yao_recursive_2022,kruse_detailed_2019}.
Now, the relationship between our covariance matrix \( \sigma \) and \( \sigma_{Q} \) is (see \cite{yao_recursive_2022}, and note that their definition of complex covariance matrix differs by a complex conjugate from ours),
\(	\sigma_{Q} = \sigma^{*} + 1/2 \)
which gives
\begin{align}\label{eq:accm}
	\amat = \sigma_{x} \pqty{ 1 - \pqty{ \sigma + 1/2 }^{-1} }
	= \sigma_x \pqty{ \sigma- 1/2 } \pqty{ \sigma+ 1/2 }^{-1} \,.
\end{align}

In conclusion, not only the adjacency matrix \( \amat \) emerges as a natural mixed-state generalization of a well-known object---the adjacency matrix of a pure gaussian state in Fock representation---but it is also worthwhile to point out that it was recently brought back to prominence by questions of recent practical interest.

Now that we have identified the double-disk representative in terms of the Fock-Bargmann kernels and Husimi function, we complete the picture and give a constructive description of which features of the pure-state disk construction are retained when trying to move to the complete picture that includes mixedness. In particular, we show that it is possible to maintain the fractional transformation / Siegel space dichotomy even when considering channels.

\section{The Double Siegel Disk}%
At this point it is useful to zoom out and state the guiding principle explicitly. What we are doing is the familiar ``up the hierarchy'' doubling of tensor order that one already uses in finite dimensions: pure states are vectors (one index), mixed states are operators (two indices), and channels are maps on operators (four indices, e.g.\ via a Choi-type identification). In other words, moving from state vectors to density matrices forces a doubling, and moving from density matrices to channels forces another doubling. This concept is neatly laid out in Ref.~\cite{yao_recursive_2022}. In our setting the same pattern is visible at the level of kernels: in the Fock--Bargmann picture, a (Gaussian) state is encoded by a Gaussian function; a mixed state is encoded by a (Gaussian) kernel with two coherent-state labels; and a channel can be encoded by a kernel with four labels (a ``super-kernel'') acting by composition/integration. The rest of the paper mirrors the standard finite-dimensional story almost verbatim: states live as a proper subset of all order-2 Gaussian kernels, channels live as a proper subset of all order-4 Gaussian super-kernels (cut out by CP/TP constraints), the unitary subgroup sits on the boundary, and single-Kraus Gaussian operations appear as the semigroup version of the unitary action.
Gaussian pure states are $K$ matrices in $\siegeldiskn$, upon which single-Kraus Gaussian operations act. We have two interpretations of these operations: either as fractional transformations induced by $\sympluscayn$ (with boundary \(\symcayn\) giving the unitaries), or as Gaussian kernels parametrized by the \emph{double Siegel disk} \(\siegeldiskN\). One may switch from one interpretation to the other with the homomorphism in \cref{eq:homosym,eq:homosyminv}, ignoring the two-to-one correspondence. Since gaussian density matrices are represented by Fock-Bargmann kernels associated to elements of a Siegel disk of dimension \(2n\), as in~\cref{eq:defstellar}, they will be a subset of the \(2n\)-dimensional disk, $\sspace \subset \siegeldiskN$.

The program is as follows. First, we define some notation that is convenient to describe matrices that we will encounter often throughout the rest of the paper. Second, we check how the pure-state disk and single-Kraus Gaussian operations are embedded in the double disk we have just identified. Third, we characterize all mixed states in the disk and discuss their Williamson decomposition. Finally, we will dedicate the rest of the paper to the discussion of channels and their description as a subset of the fractional transformations of the double disk.

\subsection{Notation}%
Here we set some notation for special types of block matrices that will be useful to describe the embedding of the groups and semigroups encountered so far into the larger set of fractional transformations we will be dealing with. Given a square matrix \(\zeta\), we will often encounter its direct sum with \(\zeta^{*}\), so we define
\begin{definition}[\(\pcc\) notation]\label{def:pcc}
\begin{align}\label{eq:notation}
	\zeta \pcc & = \zeta \oplus \zeta^{*}
	\,.
\end{align}
\end{definition}
With this in mind, ABC matrices are
\begin{definition}[ABC block form]\label{def:abc-block-form}
\begin{align}\label{eq:ctmatrix}
	T =
	\begin{pmatrix}
		\zeta    & \eta      \\
		\eta^{*} & \zeta^{*}
	\end{pmatrix}
	= \zeta\pcc + \eta \pcc \sigma_{x}
	\,,
\end{align}
\end{definition}
and satisfy $\sigma_{x} T \sigma_{x} = T^{*}$.
We can apply this type of direct sum multiple times, but we will never need to go beyond two. Namely, a double \(\bullet\pcc\) operator reads
\begin{definition}[Double-\(\pcc\) convention]\label{def:double-pcc}
\begin{align}\label{eq:doublestack}
	(\zeta\pcc)\pcc = (\zeta\oplus\zeta^{*})\pcc = \zeta \oplus \zeta^{*}\oplus \zeta^{*}\oplus \zeta
	\,.
\end{align}
\end{definition}
We define a map $T \rightarrow T \bcc$ that embeds a $2n\times 2n$ block matrix
\begin{definition}[\(\bcc\) embedding]\label{def:bcc-embedding}
\begin{align}\label{eq:blockmatrixgeneral}
	T =
	\begin{pmatrix}
		\gamma & \delta \\
		\zeta  & \eta
	\end{pmatrix}
	\,,
\end{align}
into the $4n \times 4n$ matrix:
\begin{align}\label{eq:blockmap}
	T\bcc =
	\begin{pmatrix}
		\gamma_{\pcc} & \delta_{\pcc} \\
		\zeta_{\pcc}  & \eta_{\pcc}
	\end{pmatrix}
\end{align}
\end{definition}
by applying \(\bullet\pcc\) to each block. These are matrices where each of the four blocks are block-diagonal themselves.

We will encounter different representations of the ABC symplectic group $\symcayn$. \Cref{eq:blockmap} gives one such representation when
\begin{align}\label{eq:rep1}
	\symcayn \ni S
	 & = \alpha\pcc + \beta\pcc \sigma_{x}
	\nonumber                                                \\
	\rightarrow S\bcc
	 & = (\alpha\pcc)\pcc + (\beta\pcc)\pcc \sigma_{x}
	\,,
\end{align}
where the Pauli matrix here is $4 n \times 4 n$-dimensional, twice as large as in \cref{eq:sigmablockdef}, an abuse of notation that will be ubiquitous and avoids unnecessary cluttering. Concretely, this is the lifted Pauli matrix \(\sigma_{x}^{[2n]}=\sigma_{x}\otimes I_{2n}\).

Yet another shorthand that we just used in~\cref{eq:rep1} is to reserve the symbols $\alpha$ and $\beta$ for pairs of matrices that satisfy the symplectic relations in the sense \cref{eq:alphabetaone,eq:alphabetatwo}.

\subsection{Single-Kraus Gaussian operations in \(\siegeldiskN\)}%
\label{sec:oscill-semigr-sieg}
Before identifying the set of mixed states, we show how $\siegeldiskn$ and its fractional transformations are embedded in $\siegeldiskN$. That is, we identify the set of pure states \(\siegeldiskn\) inside a larger disk, \(\siegeldiskN\), and the maps that leave \(\siegeldiskn\) invariant.
The density matrix of a pure state is $\rho_{K} = \ket{K} \bra{K}$, where \(K \in \siegeldiskn\). By writing down its FB wavefunction (or stellar function), we can show that the corresponding adjacency matrix $\amat \in \siegeldiskN$ is the block-diagonal matrix $\amat = K\pcc$. Indeed, writing down the stellar function, we immediately find
\begin{align}\label{eq:amatpure}
	\Gamma_{K} = \bra{\zeta} \rho_{K} \ket{\zeta}
	=
	e^{\zeta^{T} K \zeta + \zeta^{\dag} K^{*} \zeta^{*}}
	\,,
\end{align}
giving $\amat$ as claimed. The action of fractional transformations on pure states is also straightforward: for \(T\in\sympluscayn\), the corresponding Gaussian Kraus operator acts on pure Gaussian states by \cref{eq:actionongaussianosc}, and the induced update on \(\amat\) is
\begin{align}
	\label{eq:double_pure1}
	\amat' & =
	\bqty{\phi_{ T}(K)} \pcc
	\,, \quad
	T \in \sympluscayn
	\,,
\end{align}
i.e., single-Kraus Gaussian operations on pure states simply appear as two copies of a fractional transformation acting on each block of $\amat'$.
With the notation we have set up in \cref{eq:blockmap}, we can write this down as a fractional transformation on the state representative of the larger disk, $\amat$:
\begin{align}
	\label{eq:double_pure2}
	\amat' & =
	\phi_{T\bcc}(\amat)
	\,, \quad
	T \in \sympluscayn
	\,.
\end{align}
At this stage, this only holds when \(\amat\) is a pure state, but it is enough to describe the gaussian dynamics of pure states under the normalized oscillator semigroup (equivalently, the matrix semigroup \(\sympluscayn\)) within $\siegeldiskN$ and justifies the introduction of the embedding \(T\bcc\) in the previous section.

Now that we have settled this preliminary question, we move on to show that the double disk is large enough to accommodate all gaussian \emph{mixed states}.

\subsection{The state space}
The double disk \(\siegeldiskN\) parametrizes Gaussian kernels (oscillator-semigroup elements) in the Fock--Bargmann picture. Gaussian density matrices form a proper subset, which we will identify explicitly below and call the state space \(\sspace\subset\siegeldiskN\). This situation is fundamentally different from what we had in $\siegeldiskn$, where all the relevant states---gaussian pure states---filled the disk entirely. One way to fully identify $\sspace$ is to define it as the set of density matrices among gaussian kernels \cref{eq:gausskernels}. A simpler way to find $\sspace$ is to leverage the definition of gaussian states in terms of their complex covariance matrix and verify how the properties of covariance matrices induce properties on Siegel-disk representatives of the states. This is the approach we follow. Amongst other things, this approach allows us to translate Williamson's decomposition to the context of the disk in a straightforward manner. The resulting factorization of states can be thought of as the normal form of a disk representative. See Fig.~\ref{fig:state-space-nesting} for a schematic representation of the nested inclusion of the hierarchy of spaces that we will deal with in what follows.

The key technical bridge between the covariance-matrix picture and the double-disk picture is a fixed linear fractional \emph{change of coordinates} \(\sigma\leftrightarrow \amat\): it packages the (Hermitian) complex covariance matrix \(\sigma\) of a mixed Gaussian state into a symmetric matrix \(\amat\) living in the double disk. This coordinate change is \emph{not} the Cayley transform that relates the Siegel upper half-plane and the Siegel disk (implemented by \(\Gamma\)); rather, it is simply a convenient M\"obius-type map on covariance matrices, chosen so that affine covariance updates transport cleanly to fractional transformations on \(\amat\).
	\begin{center}
		\begin{minipage}{\linewidth}
		\centering
		\begin{tikzpicture}[every node/.style={font=\scriptsize,align=center},>=stealth]
				\node[draw,rounded corners,minimum width=6.0cm,minimum height=4.0cm,inner sep=6pt] (outer) {};
			\node[anchor=north west,xshift=4pt,yshift=-2pt,fill=white,inner sep=1pt] at (outer.north west) {\(\siegeldiskN\)};

				\node[draw,rounded corners,fill=gray!8,minimum width=5.0cm,minimum height=3.0cm,inner sep=5pt] (gamma) at (outer.center) {};
			\node[anchor=north west,xshift=4pt,yshift=-2pt,fill=gray!8,inner sep=1pt] at (gamma.north west) {\(\gammaspace\)};

				\node[draw,rounded corners,fill=gray!16,minimum width=4.0cm,minimum height=2.0cm,inner sep=4pt] (states) at (gamma.center) {};
			\node[anchor=north west,xshift=4pt,yshift=-2pt,fill=gray!16,inner sep=1pt] at (states.north west) {\(\sspace\)};

				\node[draw,dashed,rounded corners,fill=white,minimum width=2.8cm,minimum height=0.85cm,inner sep=3pt] (pure) at (states.center) {};
				\node at (pure.center) {\(\siegeldiskn\)};
			\end{tikzpicture}
				\captionof{figure}{Nesting of the doubled disk and physical subsets. The double disk \(\siegeldiskN\) is the Siegel disk in dimension \(2n\) (complex symmetric \(2n\times2n\) matrices \(\amat\) with \(I-\amat^{\dag}\amat>0\)). The ABC subset \(\gammaspace=\Delta^\Gamma_{2n}\) imposes the conjugation constraint \(\sigma_x\amat=\amat^{*}\sigma_x\). The physical state space \(\sspace=\mathcal S^\Gamma_{2n}\subset\gammaspace\) further imposes the uncertainty-principle constraint (equivalently \(W\succeq 0\) in \(\amat=\begin{psmallmatrix}K&W\\ W^{*}&K^{*}\end{psmallmatrix}\)). Pure states embed as \(\amat=K\pcc\) with \(K\in\siegeldiskn=\{K=K^{T},\ I-K^{\dag}K>0\}\).}
			\label{fig:state-space-nesting}
			\end{minipage}
		\end{center}

We first state a lemma about this fractional coordinate change for a generic \emph{covariance-like} matrix. Let \(\tau\) be Hermitian and positive definite, and assume it is ABC in the sense that \(\sigma_{x}\tau=\tau^{*}\sigma_{x}\). Define
\begin{align}
			\label{eq:aofsigma}
		\amat
		 & =
		\sigma_x
		\pqty{ \tau- 1/2 }
		\pqty{ \tau+ 1/2 }^{-1}
		\,,
	\end{align}
	which is a linear fractional transformation in our DCBA convention: \(\amat=\phi_{\Lambda}(\tau)\) with
	\begin{align}\label{eq:cov-to-disk-matrix}
		\Lambda
		&\coloneqq
		\begin{pmatrix}
			\half & 1\\
			-\half\sigma_x & \sigma_x
		\end{pmatrix},
		\qquad
		\Lambda^{\!-1}
		=
		\begin{pmatrix}
			1 & -\sigma_x\\
			\half & \half\sigma_x
		\end{pmatrix},
	\end{align}
and note that the case of a physical complex covariance matrix corresponds to \(\tau=\sigma\) (in which case \(\sigma\) additionally satisfies the uncertainty principle).
First, if \(\tau\) is ABC, then so is \(\amat\). We can show this by multiplying from both sides by \(\sigma_{x}\) in \cref{eq:aofsigma}. We get
	\begin{align}
		\sigma_x \amat \sigma_x
		 & =
	\pqty{ \tau - 1/2 }
	\pqty{ \tau + 1/2 }^{-1}
	\sigma_x
	\\ &=
	\sigma_x
	\sigma_x
	\pqty{\tau - 1/2}
	\sigma_x
	\sigma_x
	\pqty{ \tau + 1/2}^{-1}
	\sigma_x
	\\ &=
	\sigma_x
	\pqty{\tau^* - 1/2 }
	\pqty{ \tau^* + 1/2 }^{-1}
	=
			\amat^*
			\,.
		\end{align}

			We single out the ABC subset of the double disk by the defining relation \(\amat \sigma_{x} = \sigma_{x} \amat^{*}\), which we will use repeatedly to eliminate complex conjugations in later derivations:
	\begin{definition}[\(\Delta^{\Gamma}_{2n}\) (ABC subset)]\label{def:gammaspace}
	\begin{align}
		\label{eq:gammaspaceeq}
		\gammaspace = \{
		\amat \in \siegeldiskN\ |\ \amat \sigma_{x} = \sigma_{x} \amat^{*}
		\}
		\,.
	\end{align}
	\end{definition}
	From \cref{eq:aofsigma} and \(\tau = \tau^{\dag}\), it follows that \(\sigma_{x} \amat\) is hermitian:
	\begin{align}
		\label{eq:sspaceder2}
		\sigma_{x} \amat
		 & =
	\pqty{ \tau - 1/2 }
	\pqty{ \tau + 1/2 }^{-1}
	\\ &=
	\pqty{ \tau^{\dag} - 1/2 }
	\pqty{ \tau^{\dag} + 1/2 }^{-1}
	\,,
	\\
	 & =
	\pqty{
		\pqty{ \tau + 1/2 }^{-1}
		\pqty{ \tau - 1/2 }
	}^{\dag}
		= (\sigma_{x} \amat)^{\dag}
		\,,
		\end{align}
		where we used that \((\tau\pm\tfrac12)\) are polynomials in \(\tau\), hence commute with each other and with \((\tau+\tfrac12)^{-1}\) (equivalently: diagonalize \(\tau\) and verify entrywise).
		One implication is that, since, $ (\sigma_{x} \amat)^{\dag} = (\amat^{*} \sigma_{x})^{\dag} = \sigma_{x} \amat^{T} $,
			we must have $\amat=\amat^{T}$. Another implication is that the spectrum of \(\sigma_{x} \amat\) is real. Finally, positive definiteness of \(\tau\) ensures that \(\amat \in \siegeldiskN\), since
				\begin{align}\label{eq:ainsiegel}
				\amat^{*} \amat
				=
				\pqty{\sigma_{x} \amat}^{2}
			=
			\pqty{(\tau - 1/2)
				(\tau + 1/2)^{-1}}^{2}
				< 1
				\,.
			\end{align}
			Here the last inequality is in the positive-definite order; a brief diagonalization/functional-calculus justification is given in the proof of \cref{lem:tau-positive-implies-disk}.
\begin{lemma}[Fractional coordinate map sends \(\tau\) into the ABC subset]\label{lem:tau-positive-implies-disk}
		Let \(\tau\) be a \(2n\times 2n\) matrix that is Hermitian, positive definite, and ABC (\(\sigma_{x}\tau=\tau^{*}\sigma_{x}\)). Let \(\amat\) be its image under the coordinate change \cref{eq:aofsigma} (equivalently, \(\amat=\phi_{\Lambda}(\tau)\)). Then \(\amat\in\gammaspace\subset\siegeldiskN\).
\end{lemma}
\begin{proof}
	By the ABC property of \(\tau\), \cref{eq:aofsigma} implies \(\sigma_x \amat \sigma_x=\amat^{*}\), i.e. \(\amat\) satisfies the ABC relation \(\amat\sigma_{x}=\sigma_{x}\amat^{*}\).
	To show disk membership, first note that \(\amat^{*}\amat\) is Hermitian, and using \(\sigma_x^2=1\) we have \(\amat^{*}\amat=(\sigma_x\amat)^2\).
	Next, since \(\tau\) is Hermitian and positive definite, there exists a unitary \(U\) such that \(\tau=U\,\mathrm{diag}(\lambda_1,\dots,\lambda_{2n})U^{\dag}\) with \(\lambda_j>0\). Define the scalar function \(h(t)\coloneqq (t-\tfrac12)/(t+\tfrac12)\). Then
	\[
		\sigma_x\amat
		=
		(\tau-\tfrac12)(\tau+\tfrac12)^{-1}
		=
		U\,\mathrm{diag}\bigl(h(\lambda_1),\dots,h(\lambda_{2n})\bigr)U^{\dag},
	\]
	so \(\sigma_x\amat\) is Hermitian and its eigenvalues satisfy \(\abs{h(\lambda_j)}<1\). Therefore \((\sigma_x\amat)^2<1\), i.e. \(\amat^{*}\amat<1\).
	Finally, Hermiticity of \(\sigma_x\amat\) means \(\sigma_x\amat=(\sigma_x\amat)^{\dag}=\amat^{\dag}\sigma_x\). Using \(\sigma_x\amat\sigma_x=\amat^{*}\), we also have \(\amat^{\dag}=(\amat^{*})^{T}=(\sigma_x\amat\sigma_x)^{T}=\sigma_x\amat^{T}\sigma_x\). Combining these identities gives \(\sigma_x\amat=\sigma_x\amat^{T}\), hence \(\amat=\amat^{T}\). Therefore \(\amat\in\siegeldiskN\) by \cref{eq:siegeldiskdef} (with \(n\mapsto 2n\)).
	Together with the ABC relation shown at the start, this implies \(\amat\in\gammaspace\) by \cref{eq:gammaspaceeq}.
\end{proof}

The remaining physical constraint is the uncertainty principle (UP). Our goal is now to (i) transport the UP to a condition on \(\amat\in\gammaspace\), (ii) use that condition to define the double-disk \emph{state space}, and (iii) show that the image of a physical covariance matrix under the coordinate change \(\sigma\mapsto \amat\) lies in that set.

States are the subset of \(\gammaspace\) that satisfies the uncertainty principle. In terms of the complex covariance matrix \(\sigma\), the UP reads
\begin{align}\label{eq:robscr}
	\sigma - \half \sigma_{z} \ge 0
	\,.
\end{align}
To express this in disk variables, we invert the coordinate change. Writing \(X\coloneqq\sigma_{x}\amat\), we have
\begin{align}
	X
	&=
	(\sigma-\half)(\sigma+\half)^{-1},
	\\
	\sigma
	&=
	\half(1+X)(1-X)^{-1}
	\,.
\end{align}
For \(\amat\in\gammaspace\subset\siegeldiskN\), \(X\) is Hermitian and \(\operatorname{spec}(X)\subset(-1,1)\), hence \(1-X>0\) is strictly positive definite and invertible. Substituting \(\sigma=\sigma(\amat)\) into the UP \cref{eq:robscr} yields the disk-picture uncertainty principle:
\begin{align}
	\label{eq:amatupshort}
	\pqty{ 1 + \sigma_{x} \amat }
	\pqty{ 1 - \sigma_{x} \amat }^{-1}
	- \sigma_{z} \ge 0
	\,.
\end{align}
\begin{definition}[\(\sspace\) (state space)]\label{def:sspace}
	\begin{align}
		\label{eq:sspaceeq}
		\sspace \coloneqq \{
		\amat \in \gammaspace\ |\ \pqty{ 1 + \sigma_{x} \amat }
		\pqty{ 1 - \sigma_{x} \amat }^{-1}
		- \sigma_{z} \ge 0
		\}
		\,.
	\end{align}
\end{definition}
\begin{theorem}[Covariance matrices map into the state space]\label{thm:mixed-state-conditions-disk}
	Let \(\sigma\) be a (complex) covariance matrix of a Gaussian state and let \(\amat\) be its image under the coordinate change \(\sigma\mapsto\amat\) (as in \cref{eq:accm}). Then \(\amat\in\sspace\).
\end{theorem}
\begin{proof}
	Apply \cref{lem:tau-positive-implies-disk} with \(\tau=\sigma\) to obtain \(\amat\in\gammaspace\subset\siegeldiskN\), so \(1-\sigma_{x}\amat\) is invertible.
	Since \(\sigma\) is a covariance matrix, it satisfies the UP \cref{eq:robscr}. Substituting \(\sigma=\half(1+\sigma_x\amat)(1-\sigma_x\amat)^{-1}\) (the inverse of the coordinate change) into \cref{eq:robscr} gives \cref{eq:amatupshort}. Therefore \(\amat\in\sspace\) by \cref{eq:sspaceeq}.
\end{proof}
In practice, the condition \cref{eq:amatupshort} admits a simple block characterization. Write \(\amat\) in blocks as
\begin{align}
	\label{eq:ablocks}
	\amat=
	\begin{pmatrix}
		K & W \\
		W^{*} & K^{*}
	\end{pmatrix}
	\,.
\end{align}
\begin{theorem}[Mixed-state characterization theorem]\label{thm:mixed-state-characterization}
	A mixed Gaussian state corresponds to a double-disk representative \(\amat\in\sspace\). Equivalently,
	\[
		\sspace
		=
		\Bqty{
			\amat\in\gammaspace\ \big|\ W\ge 0
		},
	\]
	where \(W\) is the upper-right block in the decomposition \cref{eq:ablocks}.
\end{theorem}
\begin{proof}[Sketch of algebra]
	Write \(\amat\in\gammaspace\) in blocks as in \cref{eq:ablocks} and set \(X\coloneqq\sigma_{x}\amat\). Since \(\amat\in\siegeldiskN\), we have \(\operatorname{spec}(X)\subset(-1,1)\), so \(B\coloneqq(1-X)^{1/2}\) exists and is invertible. By a congruence with \(B\), the disk UP inequality \cref{eq:amatupshort} reduces to the positivity of an operator whose upper-block compression is \(W\), giving the equivalence \(\amat\in\sspace\iff(\amat\in\gammaspace\ \text{and}\ W\ge0)\). Finally, using the inverse coordinate map \(\sigma=\half(1+\sigma_{x}\amat)(1-\sigma_{x}\amat)^{-1}\), one checks that \(\amat\in\gammaspace\) and \(W\ge0\) imply that \(\sigma\) is a physical complex covariance matrix, hence \(\amat\) represents a Gaussian state.

	A complete derivation is given in Appendix \cref{app:thm-mixed-state-characterization}.
\end{proof}
The defining disk constraint, \(1-\amat^{\dag}\amat>0\), is geometric and is preserved by the full disk action. By contrast, the UP filter \cref{eq:amatupshort} (equivalently \(W\ge0\) inside \(\gammaspace\)) selects a distinguished cone in these coordinates. A generic fractional transformation can preserve \(\siegeldiskN\) while failing to preserve that cone, so UP is not covariant under the ambient disk action. Covariance is recovered on the unitary/symplectic subgroup, where in covariance variables one has the congruence law \(\sigma-\half\sigma_{z}\mapsto S(\sigma-\half\sigma_{z})S^{\dag}\).

\subsection{Gaussian unitaries on mixed states}

We would now like to investigate the decomposition of states as unitaries acting on thermal states, i.e., Williamson's form. What prevents us from computing this immediately is the fact that we still need to check the behavior of gaussian unitaries when they evolve mixed states, which we now do.

Under a Gaussian unitary with symplectic matrix $S \in \symcayn$, the complex covariance matrix updates as
$ \sigma'= S \sigma S^{\dag}$, inducing the transformation law on the adjacency matrix
\begin{align}
	\label{eq:atransf1}
	\amat'
	 & =
	\sigma_x
	\pqty{
		S \sigma - \half \kmat S \kmat
	}
	\pqty{
		S \sigma + \half \kmat S \kmat
	}^{-1}
	\,,
\end{align}
where we have used the symplecticity of the transformation matrix to write its inverse as $S^{-\dag} = \sigma_{z} S \sigma_{z}$. We invert~\cref{eq:aofsigma} to obtain a transformation law in terms of the initial adjacency matrix, $\amat' = \amat'(\amat)$. Writing \(\sigma\) in terms of \(\amat\),
\begin{align}
	\sigma
	 & =
	\half
	\pqty{ \sigma_x - \amat }^{-1}
	\pqty{ \sigma_x + \amat }
	\\
	& =
	\half
	\pqty{ \sigma_x + \amat^{*}}
	\pqty{ \sigma_x - \amat^{*} }^{-1}
	\,,
\end{align}
(the second equation follows from $\sigma^{\dag} = \sigma$, and is more useful here),
and substituting in~\cref{eq:atransf1}, we arrive at
\begin{align}
	\amat'
	 & =
	\sigma_x
	\pqty{
		\pqty{ S + \sigma_{z} S \sigma_{z}  }\amat^*
		+
		\pqty{ S - \sigma_{z} S \sigma_{z} } \sigma_x
	} \nonumber \\ &\quad \times
	\pqty{
		\pqty{ S - \sigma_{z} S \sigma_{z}  }\amat^*
		+
		\pqty{ S + \sigma_{z} S \sigma_{z} } \sigma_x
	}^{-1}
	\\
	 & =
	\pqty{
		\pqty{ S^{*} + \sigma_{z} S^{*} \sigma_{z}  }\amat
		+
		\pqty{ S^{*} - \sigma_{z} S^{*} \sigma_{z} } \sigma_{x}
	} \nonumber \\ &\quad \times
	\pqty{
		\pqty{ S - \sigma_{z} S \sigma_{z}  } \sigma_{x} \amat
		+
		\pqty{ S + \sigma_{z} S \sigma_{z} }
	}^{-1}
	\,,
\end{align}
	where we used $ S \sigma_{x} = \sigma_{x} S^{*}$ (because \(S\in\symcayn\)) and \( \amat \sigma_{x} = \sigma_{x} \amat^{*}\) (the defining relation of \(\gammaspace\)).
The right-hand side is a fractional transformation with respect to a \(4n\times4n\) matrix built from \(S\), which we can identify explicitly.

\begin{theorem}[Gaussian unitaries act by fractional transformations]\label{thm:unitary-action-mixed}
	Let \(S\in\symcayn\) and \(\amat\in\gammaspace\). Then the transformed adjacency matrix \(\amat'\) is given by the fractional transformation
	\begin{align}
		\label{eq:fractionalonmixed}
		\amat' & = \phi_{S \bcc} (\amat)
		\,.
	\end{align}
	In particular, \(\symcayn\) preserves \(\gammaspace\) and restricts to an action on \(\sspace\).
\end{theorem}
\begin{proof}
	The displayed computation preceding the theorem shows that the induced coordinate update (transporting \(\sigma'=S\sigma S^{\dag}\) through \(\sigma\mapsto\amat(\sigma)\)) can be rewritten as
	\[
		\amat'=(D\amat+C)(B\amat+A)^{-1}
	\]
	for an appropriate \(4n\times4n\) block matrix \(\begin{psmallmatrix}A&B\\ C&D\end{psmallmatrix}\) built from \(S\). Writing \(S=\alpha\pcc+\beta\pcc\sigma_{x}\) and using that \(\sigma_{z}\) commutes with \(\alpha\pcc\) and anticommutes with \(\beta\pcc\sigma_x\), one finds \(S\bcc=(\alpha\pcc)\pcc+(\beta\pcc)\pcc\sigma_{x}\) as in \cref{eq:rep1}, and the fraction above is precisely \(\phi_{S\bcc}(\amat)\) in the convention \cref{def:fractional-transformation}.

	Well-definedness follows because \(\amat\in\siegeldiskN\) implies \(\|\amat\|<1\), hence every denominator appearing in the derivation is invertible (equivalently, the standard Siegel-disk invertibility theorem applies since \(S\bcc\in\symcayN\)).

	To see that \(\gammaspace\) is preserved, note that \(\amat\in\gammaspace\) corresponds to an ABC covariance matrix \(\sigma\), and \(S\sigma_{x}=\sigma_{x}S^{*}\) implies \(\sigma'=S\sigma S^{\dag}\) is again ABC; the coordinate change \(\sigma\mapsto \amat(\sigma)\) preserves the ABC relation, so \(\amat'=\amat(\sigma')\in\gammaspace\).

	Finally, if \(\amat\in\sspace\), then the corresponding covariance matrix satisfies \(\sigma-\tfrac12\sigma_{z}\ge0\). Using \(S\sigma_{z}S^{\dag}=\sigma_{z}\), we get
	\[
		\sigma'-\tfrac12\sigma_{z}
		=
		S(\sigma-\tfrac12\sigma_{z})S^{\dag}\ge0,
	\]
	so \(\sigma'\) satisfies the UP and therefore \(\amat'\in\sspace\).
\end{proof}
Because to each mixed state corresponds to a double-disk element \(\amat\), the fractional transformation above describes Gaussian-unitary evolution directly on \(\siegeldiskN\). Equation \cref{eq:double_pure2} showed that, for the pure-state embedding \(\amat=K\pcc\), the pure-state update \(K\mapsto\phi_{T}(K)\) induced by \(T\in\sympluscayn\) lifts to \(\amat\mapsto\phi_{T\bcc}(\amat)\). Theorem~\ref{thm:unitary-action-mixed} shows that the same fractional form persists for mixed-state representatives \(\amat\in\sspace\) under Gaussian unitaries \(S\in\symcayn\), namely \(\amat'=\phi_{S\bcc}(\amat)\).

\subsection{Normal-mode decomposition in the disk}%
\label{sec:normalmodedec}

To write down a normal-mode decomposition of a state in the disk, in principle, one would need to consider a state $\amat \in \sspace$ and use its spectral properties to show that it is decomposable in some notable way. Since fractional transformations are the natural automorphism of the disk, it is reasonable to have the expectation that whatever the explicit form of such a disk-normal-mode decomposition is, it should be a fractional transformation.
Geometrically, \(\sspace\) carries an action of \(\symcayn\) by disk automorphisms, so a ``normal form'' amounts to moving any \(\amat\) along its group orbit to a simple representative (thermal) and interpreting the remaining freedom via the stabilizer.

There is a shortcut, however. Because for complex covariance matrices we already have the Williamson form, we can simply inspect its decomposition in the Siegel disk by using the parametrization of \(\sigma\) by one temperature matrix, $\nu\pcc$, and an ABC, symplectic matrix $S \in \symcayn$ [\cref{eq:complex_covariance_matrix}].
The normal-mode decomposition in the disk is then defined through the map between covariance matrices and adjacency matrices in \cref{eq:aofsigma}. In practice, we parametrize the covariance matrix as $\sigma = \half S \nu_{\pcc} S^{\dag}$ and check the image in the hope that the decomposed nature survives in a suggestive way. Indeed it does.

The adjacency matrix of a thermal state is obtained from~\cref{eq:aofsigma},
\begin{align}\label{eq:thermalesdr}
	\amat_{\xi} = \xi\pcc \sigma_{x}
	\,,
\end{align}
where we set
	\begin{align}
		\label{eq:normalmodedec}
		\xi & = (\nu-1)(\nu + 1)^ {-1}
		\,.
	\end{align}
	Since \(\nu\) is diagonal with \(\nu_i>0\), this fractional coordinate can equivalently be written entrywise as
	\(\xi=\tanh\!\pqty{\tfrac12 \log \nu}\).
	In particular, \(-1<\xi<1\), and for physical temperatures \(\nu\ge 1\) we have \(0\le \xi<1\).
	The matrix $\sigma_{x} \amat_{\xi}$ is diagonal, and \cref{eq:normalmodedec} makes it clear that with no further restrictions, its spectrum is $[-1,1]$.
	The uncertainty principle further narrows it down.
We know that states have \(\nu \ge 1\), but let us check with the disk version of the UP instead, \cref{eq:amatupshort}. For thermal states, \(W = \xi\), and we obtain \(\xi \ge 0\), which is equivalent to \(\nu \ge 1\).

Since we have shown that gaussian unitaries act on any state via a fractional transformation with respect to a symplectic matrix \(S\), then \cref{eq:fractionalonmixed} holds, and the normal mode decomposition of the state is simply \(\phi_{S}\) on \(\amat_{\xi}\), i.e.,
\begin{align}
	\label{eq:normalmodedec2}
	\amat = \phi_{S\bcc}(\amat_{\xi})
	\,,
\end{align}
with $\amat_{\xi}$ given by~\cref{eq:thermalesdr}, and $S \in \symcayn$.
The embedding $S \rightarrow S \bcc$ is defined in \cref{eq:blockmap}.
A more direct way to obtain the same answer (that involves more algebra) is to simply plug \(\sigma = \tfrac{1}{2} S \nu\pcc S^{\dag}\) into \cref{eq:aofsigma} and check block by block.

\begin{corollary}[Williamson decomposition in the disk]\label{cor:williamson-disk}
	For any adjacency matrix \(\amat\in\sspace\), there exists \(S\in\symcayn\) and a thermal representative \(\amat_{\xi}\) of the form \cref{eq:thermalesdr} such that
	\(\amat=\phi_{S\bcc}(\amat_{\xi})\).
\end{corollary}
\begin{proof}
	Let \(\amat\in\sspace\) and let \(\sigma=\sigma(\amat)\) be its physical covariance matrix. By Williamson's theorem there exist \(S\in\symcayn\) and \(\nu\ge 1\) such that \(\sigma=\tfrac12 S\nu\pcc S^{\dag}\). Setting \(\sigma_{\mathrm{th}}\coloneqq\tfrac12\nu\pcc\) and \(\amat_{\xi}\coloneqq \amat(\sigma_{\mathrm{th}})\), the mixed-state unitary action \cref{thm:unitary-action-mixed} gives \(\amat=\phi_{S\bcc}(\amat_{\xi})\), and \(\amat_{\xi}\) has the thermal form \cref{eq:thermalesdr} by direct evaluation of \cref{eq:aofsigma} on \(\sigma_{\mathrm{th}}\). See, e.g., Ref.~\cite{serafini_quantum_2017} for Williamson's theorem.
\end{proof}

\section{Gaussian operations and channels in the Siegel Disk}
\label{sec:gaussian-channels-disk}
We now turn from single-Kraus Gaussian operations (semigroup elements) to general Gaussian operations and, in particular, deterministic Gaussian channels. In the disk picture, a deterministic Gaussian channel should act on the double-disk representatives \(\amat\) by a linear fractional transformation while preserving the physical state space \(\sspace\subset\siegeldiskN\). The inclusions and embeddings used in this section are summarized in Fig.~\ref{fig:big-picture-disk}.

	The natural transformations acting on \(\siegeldiskN\) are the disk-preserving fractional transformations in dimension \(4n\), i.e.\ the Siegel-disk preserving semigroup \(\sympluscayN\) defined exactly as in \cref{def:sympluscayn} but with \(n\mapsto 2n\) (so elements are \(4n\times4n\) matrices acting on the \(2n\)-mode disk by linear fractional transformations).
	\begin{center}
		\begin{minipage}{\linewidth}
			\centering
			\begin{tikzpicture}[>=stealth,x=1cm,y=1cm,
				box/.style={draw=black,rounded corners=8pt,line width=0.35pt,fill=gray!6,fill opacity=0.40,text opacity=1,inner sep=7pt,font=\scriptsize,align=left,text width=7.6cm,minimum height=1.15cm},
				arrow/.style={->,line width=0.35pt}
			]
				\node[box] (coord) at (0,3.05) {\textbf{Coordinates:} \(\sigma\ \xleftrightarrow{\ \phi_{\Lambda}\ }\ \amat\in\siegeldiskN\).};
				\node[box] (phys) at (0,1.65) {\textbf{Physical subset:} \(\gammaspace\) (ABC) and UP \(\Rightarrow\ \sspace\subset\siegeldiskN\).};
				\node[box] (dyn) at (0,0.10) {\textbf{Dynamics:} updates are fractional maps \(\amat'=\phi_{E}(\amat)\).\\[-1pt]
					Unitaries: \(S\in\symcayn\mapsto S\bcc\) and \(\amat'=\phi_{S\bcc}(\amat)\).\\[-1pt]
					Single-Kraus: \(T\in\sympluscayn\mapsto T\bcc\) and \(\amat'=\phi_{T\bcc}(\amat)\).\\[-1pt]
					Channels: \((X,Y)\mapsto\bar E(X,Y)\) and \(\amat'=\phi_{\bar E}(\amat)\).};
				\node[box] (comp) at (0,-2.05) {\textbf{Composition:} \(\phi_{E_2}\circ\phi_{E_1}=\phi_{E_2E_1}\) (matrix multiplication).};

				\draw[arrow] (coord.south) -- (phys.north);
				\draw[arrow] (phys.south) -- (dyn.north);
				\draw[arrow] (dyn.south) -- (comp.north);
			\end{tikzpicture}
			\captionof{figure}{Program summary (disk viewpoint). We use a fixed fractional coordinate change \(\sigma\xleftrightarrow{\phi_{\Lambda}}\amat\) to represent mixed Gaussian states as \(\amat\in\siegeldiskN\), then restrict to the physical state space \(\sspace\subset\siegeldiskN\) by imposing ABC symmetry and the UP. Gaussian dynamics is represented by fractional transformations \(\amat'=\phi_{E}(\amat)\) with acting matrices \(E\) obtained by embedding (i) unitaries \(S\in\symcayn\), (ii) single-Kraus operations \(T\in\sympluscayn\), or (iii) deterministic channels via \((X,Y)\mapsto\bar E(X,Y)\). Composition corresponds to multiplication of the acting matrices.}
			\label{fig:big-picture-disk}
		\end{minipage}
	\end{center}
	The remainder of this section explains how far one can get by viewing channels abstractly as transformations on disk representatives, and then connects that viewpoint to the standard covariance-matrix parametrization \((X,Y)\), which provides the concrete description we will use for channel dynamics.
			
\subsection{A geometric viewpoint on channels}\label{sec:channels-fractional-maps}%
		
	In the Fock--Bargmann picture, a Gaussian operation/channel can be represented by a Gaussian super-kernel acting on kernels by composition/integration. In the disk language developed here, this amounts to specifying a transformation acting on mixed-state representatives \(\amat\in\sspace\subset\siegeldiskN\). At this level of abstraction, this is an ``identification'' in the same sense that one may say ``a (quantum) channel is a rank-4 tensor'': it locates operations/channels inside a very large space of transformations, without yet giving an intrinsic, practical characterization of which ones are physical.

	\begin{definition}[Gaussian operations and channels]\label{def:gchannels}
		We write \(\gchannels\) for the set of \(n\)-mode \emph{Gaussian operations}, i.e.\ completely positive, trace-nonincreasing linear maps on trace-class operators that map Gaussian states to Gaussian states. Following~\cite{serafini_quantum_2017}, we refer to \(\gchannels\) as \emph{non-deterministic} in general, and we write \(\gchannelsTP\subseteq \gchannels\) for the trace-preserving (deterministic) subclass (Gaussian channels).
	\end{definition}
	In \cref{sec:channel-embedding} we use the standard \((X,Y)\) parametrization of Gaussian channels to associate to each \(\Phi\in\gchannelsTP\) a matrix \(\bar E(\Phi)\in\sympluscayN\) whose disk fractional action reproduces the channel update on representatives. By (mild) abuse of language, we will often identify \(\Phi\) with \(\bar E(\Phi)\) and thus view the embedded channel class as a concrete subset \(\gchannelsTP\subseteq\sympluscayN\). In particular, the embedded class obeys the two coarse geometric constraints introduced below (state-space preservation on \(\sspace\) and ABC preservation on \(\gammaspace\)); see Fig.~\ref{fig:channel-constraints}. The embedded single-Kraus semigroup \(\mathcal T\coloneqq(\sympluscayn)\bcc\) is a distinguished subclass of Gaussian operations, \(\mathcal T\subseteq\gchannels\), but in general it is not trace preserving; correspondingly, for endomaps on \(n\) modes one has \(\mathcal T\cap\gchannelsTP=\mathcal U\coloneqq(\symcayn)\bcc\) (embedded unitaries).

		A natural necessary condition for physicality in this setting is that the induced update on representatives preserve the state space, i.e.\ \(\sspace\to\sspace\) (Fig.~\ref{fig:big-picture-disk}, top-right arrow). This motivates the following \emph{geometric} class:
			
		\begin{definition}[State-space preserving class \(\mathcal E\)]\label{def:channel-set}
		\begin{align}\label{eq:gauschangeneral}
			\mathcal E = \Bqty{E \in \sympluscayN \ \text{s.t.} \ \phi_{E}(\amat) \in \sspace, \ \forall \amat \in \sspace}
		\,.
	\end{align}
	\end{definition}
	This definition is intentionally coarse: it does not, by itself, encode complete positivity/linearity at the operator level, and turning the condition \(\phi_E(\sspace)\subseteq\sspace\) into useful constraints on the blocks of \(E\) is not straightforward. For our purposes, it is therefore not worth pursuing a full intrinsic characterization of physical channels directly inside \(\sympluscayN\). Instead, in \cref{sec:channel-embedding} we use the standard \((X,Y)\) parametrization of Gaussian channels to obtain an explicit embedding that yields the disk-picture update rule we need.

	That said, it is still useful to understand some \emph{structural} geometric constraints that are compatible with our coordinates. The simplest is to require that \(E\) preserve the ABC subset \(\gammaspace\), which admits a clean block-level characterization and provides a natural intermediate class (Fig.~\ref{fig:channel-constraints}).
	Consider the subset of ABC-preserving fractional transformations of \(\sympluscayN\),
	\begin{definition}[ABC-preserving maps \(\mathcal E^{\Gamma}\)]\label{def:abc-preserving-maps}
	\begin{align}\label{eq:gauschangeneralGamma}
		\mathcal E^{\Gamma} = \Bqty{E \in \sympluscayN \ \text{s.t.} \ \phi_{E}(\amat) \in \gammaspace, \ \forall \amat \in \gammaspace}
	\,.
\end{align}
\end{definition}
Because \(\sspace\subset\gammaspace\), any channel must preserve the ABC condition \emph{at least on the state space}: if \(\amat\in\sspace\), then \(\phi_{E}(\amat)\in\sspace\subset\gammaspace\). It is therefore natural to also consider the stronger requirement that \(E\) preserve \(\gammaspace\) on all of \(\gammaspace\), which is exactly the set \(\mathcal E^{\Gamma}\) above. The relationship between \(\mathcal E\) and \(\mathcal E^{\Gamma}\) is not automatic in general, and we will mostly work within the ABC-preserving class because it admits a clean block-level characterization.
	With the notation of \cref{def:channel-set,def:abc-preserving-maps}, the embedded deterministic channel class satisfies \(\gchannelsTP\subseteq \mathcal E\cap \mathcal E^{\Gamma}\) in the schematic sense of Fig.~\ref{fig:channel-constraints}.

	\begin{center}
		\begin{minipage}{\linewidth}
				\centering
				\begin{tikzpicture}[x=1cm,y=1cm,
					every node/.style={font=\scriptsize,align=center},
					frame/.style={draw=black,rounded corners=8pt,line width=0.35pt},
					slabA/.style={draw=black,rounded corners=10pt,line width=0.35pt,fill=gray!10,fill opacity=0.35},
					slabB/.style={draw=black,rounded corners=10pt,line width=0.35pt,fill=gray!20,fill opacity=0.35},
					sq/.style={draw=black,rounded corners=8pt,line width=0.35pt}
				]
					\draw[frame] (-3.35,-2.95) rectangle (3.35,2.95);
					\node[anchor=south] at (0,2.95) {\(\sympluscayN\)};

					\filldraw[slabA] (-3.10,-1.45) rectangle (3.10,2.65);
					\filldraw[slabB] (-3.10,-2.65) rectangle (3.10,1.45);
					\node[anchor=north west] at (-3.02,2.58) {\(\mathcal E\)};
					\node[anchor=south west] at (-3.02,-2.58) {\(\mathcal E^{\Gamma}\)};
					\node[anchor=south] at (0,1.45) {\(\mathcal E\cap\mathcal E^{\Gamma}\)};

					\filldraw[sq,fill=gray!26,fill opacity=0.55]
						(-1.55,-1.25) rectangle (1.55,1.15);
					\node at (0,0.92) {\(\gchannels\)};

					\filldraw[sq,fill=gray!38,fill opacity=0.55]
						(-1.22,-0.20) rectangle (1.22,0.78);
					\filldraw[sq,fill=gray!38,fill opacity=0.55]
						(-1.22,-0.78) rectangle (1.22,0.20);
					\begin{scope}
						\path[clip,rounded corners=8pt] (-1.22,-0.20) rectangle (1.22,0.78);
						\filldraw[sq,fill=gray!62,fill opacity=0.60]
							(-1.22,-0.78) rectangle (1.22,0.20);
					\end{scope}
					\node at (0,0.42) {\(\mathcal T\)};
					\node at (0,-0.42) {\(\gchannelsTP\)};
					\node at (0,0.00) {\(\mathcal U\)};
				\end{tikzpicture}
			\captionof{figure}{Constraint hierarchy for fractional transformations on the double disk (schematic). The ambient semigroup \(\sympluscayN\) acts on \(\siegeldiskN\) by linear fractional transformations. Two coarse geometric classes are \(\mathcal E\) (state-space preserving on \(\sspace\)) and \(\mathcal E^{\Gamma}\) (ABC-preserving on \(\gammaspace\)). Inside their overlap we indicate the (schematic) physically meaningful class \(\gchannels\) of Gaussian operations, together with two notable subclasses that we explicitly realize in this paper: \(\mathcal T\coloneqq(\sympluscayn)\bcc\) (embedded single-Kraus/oscillator semigroup) and \(\gchannelsTP\) (embedded deterministic Gaussian channels). Their intersection is \(\mathcal U\coloneqq(\symcayn)\bcc=\mathcal T\cap\gchannelsTP\) (embedded Gaussian unitaries).}
			\label{fig:channel-constraints}
		\end{minipage}
	\end{center}

		\begin{center}
			\begin{minipage}{\linewidth}
					\centering
					\begin{tikzpicture}[>=stealth,
						box/.style={inner sep=1pt,minimum height=1.05cm,font=\scriptsize,align=center},
						arrow/.style={->,line width=0.35pt},
						lbl/.style={font=\scriptsize,fill=white,inner sep=1.2pt,rounded corners=2pt}
						]
						\matrix (m) [matrix of nodes,row sep=1.30cm,column sep=0.95cm] {
							\node[box,text width=2.95cm] (sigma) {\(\sigma\)}; &
							\node[box,text width=3.20cm] (sigmap) {\(\sigma'\)}; \\
							\node[box,text width=2.95cm] (A) {\(\amat=\phi_{\Lambda}(\sigma)\in\sspace\)}; &
							\node[box,text width=3.20cm] (Ap) {\(\amat'=\phi_{\Lambda}(\sigma')\in\sspace\)}; \\
						};

					\node[font=\scriptsize,align=center] at ($(sigma.north)!0.5!(sigmap.north)+(0,0.62)$)
						{\(\Phi\in\gchannelsTP\)\\\((X,Y)\)};

						\draw[arrow] (sigma) -- node[midway,below=5pt,lbl]{\(\sigma\mapsto X\sigma X^{\dag}+Y\)} (sigmap);
						\draw[arrow] (A) -- node[midway,below=5pt,lbl]{\(\substack{\amat\mapsto\\ \phi_{\bar E(X,Y)}(\amat)}\)} (Ap);
						\draw[arrow] (sigma) -- node[midway,left,lbl]{\(\phi_{\Lambda}\)} (A);
						\draw[arrow] (sigmap) -- node[midway,right,lbl]{\(\phi_{\Lambda}\)} (Ap);
					\end{tikzpicture}
					\captionof{figure}{Deterministic channels act on \(\sspace\) by disk fractional transformations (schematic). The embedding theorem asserts that this square commutes on the state space: for \(\Phi\in\gchannelsTP\) with parameters \((X,Y)\) and embedded acting matrix \(\bar E=\bar E(X,Y)\), one has \(\phi_{\Lambda}(X\sigma X^{\dag}+Y)=\phi_{\bar E}(\phi_{\Lambda}(\sigma))\) for \(\sigma\) representing a state (equivalently, \(\amat=\phi_{\Lambda}(\sigma)\in\sspace\)).}
					\label{fig:gtp-action-sspace}
				\end{minipage}
			\end{center}

	In particular, for \(\amat\in\gammaspace\) we want $\phi_{E}(\amat)$ to satisfy~\cref{eq:gammaspaceeq}, with $E \in \sympluscayN$:
	\begin{align}\label{eq:physicalmaps}
		\sigma_{x} \phi^{*}_{E}(\amat) \sigma_{x} = \phi_{E}(\amat)
	\,.
\end{align}
We can identify these maps by inspecting the transformation explicitly in terms of its blocks. Let
\begin{align}\label{eq:sabcd}
	E =
	\begin{pmatrix}
		\gamma & \delta \\
		\zeta  & \omega
	\end{pmatrix}
	\,,
\end{align}
where the blocks are $2 n \times 2 n$ complex dimensional matrices. The condition in \cref{eq:physicalmaps}, in terms of the blocks of $E$, reads
\begin{align}\label{eq:physicalsderivation}
	\nonumber
	\sigma_{x} \phi^{*}_{E} (\amat) \sigma_{x}
	 & =
	\pqty{\sigma_{x} \gamma^{*} \sigma_{x} \amat + \sigma_{x} \delta^{*} \sigma_{x}}
	\\ &\quad\times
	\pqty{ \sigma_{x} \zeta^{*} \sigma_{x} \amat + \sigma_{x} \omega^{*} \sigma_{x} }^{-1}
	\\
	&=
	\phi_{E} (\amat)
	\,,
\end{align}
and the equals sign holds if each block, $\gamma$, $\delta$, $\zeta$, and $\omega$ is an ABC matrix, or, in other words, if it belongs to
\begin{align}\label{eq:blockscay}
	\mathcal{E}^{\Gamma} =
	\Bqty{
		E \in \sympluscayN\ |\
		E ( \sigma_{x} )\pcc
		=
		( \sigma_{x} )\pcc E^{*}
	}
	\,.
\end{align}
Explicitly, these are \(4n\)-dimensional matrices parametrized by \(8\) \(n\)-dimensional matrices as
\begin{align} \label{eq:chanparametrizationgeneral}
	E & = 	\begin{pmatrix}
			       \mu_{1}       & \mu_{2}       & \nu_{1}        & \nu_{2}        \\
		       \mu_{2}^{*}   & \mu^{*}_{1}   & \nu_{2}^{*}    & \nu_{1}^{*}    \\
		       \zeta_{1}^{*} & \zeta_{2}^{*} & \omega_{1}^{*} & \omega_{2}^{*} \\
		       \zeta_{2}     & \zeta_{1}     & \omega_{2}     & \omega_{1}     \\
	       \end{pmatrix}
	\,,
\end{align}
All deterministic Gaussian channels, as well as a large set of unphysical transformations, have this form.

Belonging to this set does not imply that \(E\) itself is ABC. If, in addition, \(E\) satisfies the ABC relation \(E(\sigma_{x})^{[2n]}=(\sigma_{x})^{[2n]}E^{*}\) (equivalently: \(E\) is an ABC-symplectic matrix), then the disk-preserving inequality \(E^{\dag}\sigma_{z}E\ge\sigma_{z}\) saturates to \(E^{\dag}\sigma_{z}E=\sigma_{z}\), hence \(E\in\symcayN\), the boundary group of \(\sympluscayN\). This is another set of fractional transformations of interest, which still include unphysical transformations. The elements of \(\symcayN\) can be parametrized by four $n \times n$ matrices as
\begin{align}\label{eq:chanparametrization}
	\bar S & =
	\pqty{(\mu_{1})_{\pcc} + (\mu_{2})_{\pcc} \sigma_{x}}\pcc
	\nonumber        \\
	       & \quad +
	\pqty{
		(\nu_{1})\pcc + (\nu_{2})\pcc \sigma_{x})
	}\pcc \sigma_{x}
	\,,
\end{align}
or, more explicitly,
\begin{align} \label{eq:chanparametrizationexplicit}
	\bar S & =
	\begin{pmatrix}
		\mu_{1}     & \mu_{2}     & \nu_{1}     & \nu_{2}     \\
		\mu_{2}^{*} & \mu^{*}_{1} & \nu_{2}^{*} & \nu_{1}^{*} \\
		\nu_{1}^{*} & \nu_{2}^{*} & \mu_{1}^{*} & \mu_{2}^{*} \\
		\nu_{2}     & \nu_{1}     & \mu_{2}     & \mu_{1}     \\
	\end{pmatrix}
	\,,
\end{align}
where the \(2n\)-dimensional blocks
$\alpha = (\mu_{1})\pcc + (\mu_{2})\pcc \sigma_{x}$
and
$\beta = (\nu_{1})\pcc + (\nu_{2}) \sigma_{x}$
satisfy~\cref{eq:alphabetaone,eq:alphabetatwo}.
This set is a genuine group (as opposed to a semigroup): the defining ABC relation \(\bar S(\sigma_x)^{[2n]}=(\sigma_x)^{[2n]}\bar S^{*}\) is preserved under products,
\[
(\bar S\bar S')(\sigma_x)^{[2n]}
=\bar S(\sigma_x)^{[2n]}\bar S'^{*}
=(\sigma_x)^{[2n]}(\bar S\bar S')^{*},
\]
and taking inverses and complex conjugates shows \(\bar S^{-1}\) obeys the same relation.
This parametrization makes it easier to inspect subsets where some of the blocks are zero.

Gaussian unitaries, $\symcayn$, are their subgroup obtained with the choice $\mu_{2} = \nu_{2} = 0$. The symplectic conditions then give $\mu_{1} = \alpha$ and $\nu_{1} = \beta$, where $\alpha$ and $\beta$ are a symplectic pair of \(n\)-dimensional blocks, and the parametrization above gives exactly the matrix $S\bcc$, with \(S \in \symcayn\), which we already showed is a representation of Gaussian unitaries.

We can single out the embedding of single-Kraus Gaussian operations by looking for those automorphisms of $\sspace$ that map all pure states to pure states. Since the adjacency matrix of a pure state is diagonal, $K\pcc$, we need to look for those fractional transformations that preserve block-diagonal form. Requiring this property on \cref{eq:chanparametrizationgeneral} implies that the elements indexed by \(2\) are zero. That is,
$\phi_{E} (K\pcc) = K'\pcc$ identifies the semigroup parametrized as
\begin{align} \label{eq:physicalgroup2}
	E & =
	\pqty{\gamma\pcc \oplus \omega\pcc}
	+
	\pqty{\delta\pcc \oplus \zeta\pcc} \sigma_{x}
	\,,
\end{align}
i.e., \(E = T\bcc\), where \(T \in \sympluscayn\), which is precisely the embedding of these single-Kraus Gaussian operations in \(\siegeldiskN\) that we described in \cref{eq:double_pure2}.

\subsection{Channel embedding for deterministic Gaussian channels}\label{sec:channel-embedding}

So far, we have sidestepped the issue of requiring that deterministic channels preserve the uncertainty principle. We did so by first defining an interesting set of (not necessarily physical) operations that preserve \(\gammaspace\) instead. We now complete the picture by taking a well-known explicit parametrization of all deterministic Gaussian channels and map it to the double disk. This will provide a disk parametrization of all possible deterministic Gaussian channels, as well as a relationship between the traditional picture of Gaussian channels and its double-disk counterpart.

The \((X,Y)\) parametrization we use here is the standard description of \emph{deterministic} (trace-preserving) Gaussian channels~\cite{serafini_quantum_2017}. The most general gaussian channel is characterized by its action on the real covariance matrix, $\Sigma$, through the map~\cite{serafini_quantum_2017}
\begin{align}\label{eq:gauschansreal}
	\Sigma' = X_{r} \Sigma X_{r}^{T} + Y_{r}
	\,,
\end{align}
where $X_{r}$ and $Y_{r}$ are real $2 n \times 2 n$ matrices satisfying
\begin{align}\label{eq:gauschanrealconditions}
	Y_{r} - \sigma_{y} \geq - X_{r} \sigma_{y} X_{r}^{T}
	\,.
\end{align}
We find the equivalent statement for complex covariance matrices by multiplying the above by $\Gamma$ and its adjoint, which gives
\begin{align}\label{eq:chaneq}
	\sigma' = X \sigma X^{\dag} + Y
	\,,
\end{align}
where $X$ and $Y$ are $2n \times 2n$ ABC real matrices. In addition, $Y$ can be chosen such that $Y = Y^{\dag}$.
The inequality ensuring that the channel does not violate the uncertainty principle is (using $\Gamma \sigma_{y} \Gamma^{\dag} = - \sigma_{z}$)
\begin{align}\label{eq:chanineq}
	X \sigma_{z} X^{\dag} + Y \ge \sigma_{z}
	\,.
\end{align}
Let us verify that \cref{eq:chaneq,eq:chanineq} define a fractional transformation on adjacency matrices representing states in the double disk $\amat \in \sspace$. The adjacency matrix after the channel in \cref{eq:gauschansreal} is obtained by applying the coordinate change \(\sigma\mapsto\amat(\sigma)\) to \cref{eq:chaneq},
	\begin{align}
		\amat'
		 & =
		\sigma_x
		\pqty{X \sigma X^{\dag} + Y - \half }
		\pqty{ X \sigma X^{\dag} + Y + \half }^{-1}
		\,.
	\end{align}
		Transporting \(\sigma\mapsto X\sigma X^{\dag}+Y\) to an update on \(\amat\) proceeds exactly as in the unitary special case \cref{thm:unitary-action-mixed}: substitute the inverse coordinate map \(\sigma=\sigma(\amat)\), use the ABC relations to eliminate starred quantities, and collect the result into a single \(4n\times4n\) block matrix acting by our DCBA fractional convention \cref{def:fractional-transformation}. Since the final expression involves \(X^{-T}\) and \(X^{-\dag}\), we assume throughout that \(X\) is invertible. The algebraic reduction is given in \cref{app:thm-channel-embedding}; the resulting embedded matrix can be written in \(2\times2\) block form as
	\begin{align}\label{eq:sympexp}
		\bar E
		 &=
		\half
		\begin{pmatrix}
			E_{11} & E_{12} \\
			E_{21} & E_{22}
		\end{pmatrix}
		\,,
		\\
		E_{11}
		 &\coloneqq
		X + X^{-\dag} + 2 Y X^{-\dag}
		\,,
		\nonumber\\
		E_{12}
		 &\coloneqq
		\pqty{X - X^{-\dag} - 2 Y X^{-\dag}} \sigma_{x}
		\,,
		\nonumber\\
		E_{21}
		 &\coloneqq
		\pqty{X^{*} - X^{-T} + 2 Y^{*} X^{-T}} \sigma_{x}
		\,,
		\nonumber\\
		E_{22}
		 &\coloneqq
		X^{*} + X^{-T} - 2 Y^{*} X^{-T}
		\,.
	\end{align}
		For comparison with the recursive/Fock-space formulation of Gaussian channels, Ref.~\cite{yao_recursive_2022} packages a channel \((X,Y)\) as a (Gaussian) Choi state. In their notation, Eq.~(B18) defines the Husimi covariance matrix \(\xi\) (the \(Q\)-covariance of the vacuum output), while Eq.~(B12) gives the corresponding Choi-state \(A\)-matrix \(A_{\Phi}=P_{2M}R\!\begin{psmallmatrix}1-\xi^{-1}&\xi^{-1}X\\ X^{T}\xi^{-1}&1-X^{T}\xi^{-1}X\end{psmallmatrix}\!R^{\dag}\). In our embedding, the same \(\xi\) combination is already encoded in the \(\bar E\) blocks: one may rewrite \(E_{11}=2\xi X^{-\dag}\) (equivalently \(\xi=\half E_{11}X^{\dag}\)). Thus \(A_{\Phi}\) is a point in the doubled disk (channel-as-state), whereas \(\bar E\) is an acting matrix whose fractional action gives the disk update \(\amat\mapsto\phi_{\bar E}(\amat)\) (channel-as-dynamics).
	The matrix \(\bar E\) is the double-disk representative of the channel: it encodes the full \((X,Y)\) data in a single \(4n\times4n\) matrix whose fractional action reproduces the channel on \(\sspace\).

	\begin{theorem}[Channel embedding theorem for deterministic Gaussian channels (disk fractional transform)]\label{thm:channel-embedding}
		Let \((X,Y)\) satisfy \cref{eq:chanineq} with \(X\) invertible, and define \(\bar E\) by \cref{eq:sympexp}. Then the induced action on double-disk representatives \(\amat\in\sspace\) is the fractional transformation \(\amat'=\phi_{\bar E}(\amat)\), and this map is well-defined on \(\sspace\) and preserves \(\sspace\).
		Moreover, composition of channels corresponds to multiplication of the corresponding \(\bar E\) matrices.
	\end{theorem}
	\begin{proof}
		The identification of the coordinate update \(\amat(\sigma')\), with \(\sigma'=X\sigma X^{\dag}+Y\), as the fractional transformation \(\phi_{\bar E}\) (in our DCBA convention \cref{def:fractional-transformation}) is a direct algebraic transport through the coordinate change \(\sigma\leftrightarrow \amat\), completely analogous to the unitary case \cref{thm:unitary-action-mixed}; we give the reduction in \cref{app:thm-channel-embedding}.

		To see that the map is well-defined on \(\sspace\) and preserves \(\sspace\), let \(\amat\in\sspace\) and let \(\sigma=\sigma(\amat)\) be its (physical) covariance matrix. Then \(\sigma'=X\sigma X^{\dag}+Y\) is again a physical covariance matrix by \cref{eq:chanineq}. In particular, \(\sigma'+\half>0\), so the coordinate map defining \(\amat(\sigma')\) is well-defined. Since \(\amat'=\amat(\sigma')=\phi_{\bar E}(\amat)\), this also implies that \(\phi_{\bar E}(\amat)\) is defined and \(\amat'\in\sspace\) (by \cref{thm:mixed-state-conditions-disk}).

		Finally, composition is inherited from composition of covariance updates and the coordinate transport, and equivalently from the general composition law \(\phi_{E_2}\circ\phi_{E_1}=\phi_{E_2E_1}\) whenever both sides are defined.
	\end{proof}

	\begin{proof}[Alternative proof (fractional-composition)]
				The key observation is that both (i) the covariance update \(\sigma\mapsto X\sigma X^{\dag}+Y\) and (ii) the coordinate map \(\sigma\mapsto \amat(\sigma)\) are themselves linear fractional transformations, so the disk-picture channel update is obtained by conjugation and the composition law \cref{eq:compositionlaw}.

				First note that the coordinate map \cref{eq:aofsigma} is a linear fractional transformation: in our DCBA convention,
				\[
					\amat=\phi_{\Lambda}(\sigma),
					\]
					where \(\Lambda\) (and \(\Lambda^{-1}\)) are the fixed \(4n\times4n\) matrices from \cref{eq:cov-to-disk-matrix}.

				Next, in our DCBA convention \cref{def:fractional-transformation}, the affine covariance update is the fractional transformation
				\begin{align}
					X\sigma X^{\dag}+Y
				&=
				\phi_{L_{Y}B_{X}}(\sigma),
				\\
				L_{Y}
				&\coloneqq
				\begin{pmatrix}
					1 & 0 \\
					Y & 1
				\end{pmatrix},
				\\
				B_{X}
				&\coloneqq
				\begin{pmatrix}
					X^{-\dag} & 0 \\
					0 & X
				\end{pmatrix},
			\end{align}
		as is verified by direct substitution into \(\phi_T(Z)=(DZ+C)(BZ+A)^{-1}\).

			Combining these two observations and using \cref{eq:compositionlaw} gives
			\begin{align}
				\amat'
				=
					\phi_{\Lambda}(\sigma')
					=
					\phi_{\Lambda}\circ\phi_{L_{Y}B_{X}}(\sigma)
					=
					\phi_{\Lambda L_{Y}B_{X}\Lambda^{-1}}(\amat).
				\end{align}
				Defining \(\bar E\coloneqq \Lambda L_{Y}B_{X}\Lambda^{-1}\), a short block multiplication yields \(\bar E\) in the form \cref{eq:sympexp} (using only the ABC relations \(\sigma_xX=X^{*}\sigma_x\) and \(\sigma_xY=Y^{*}\sigma_x\) to rewrite the lower blocks), so \(\amat'=\phi_{\bar E}(\amat)\) as claimed.

			In the unitary (Gaussian) special case \(Y=0\) and \(X=S\in\symcayn\), this reduces to \(\bar E=\Lambda B_{S}\Lambda^{-1}\). Using the symplectic identity \(S^{-\dag}=\sigma_z S\sigma_z\) in \cref{eq:sympexp} recovers exactly the previously-defined embedding \(S\mapsto S\bcc\) and hence \cref{thm:unitary-action-mixed}.
		\end{proof}

	For reference, we present the disk-picture channel update in compact form:
	\begin{align}\label{eq:channelfractransf}
		\amat' = \phi_{\bar E}(\amat)
		\,.
	\end{align}

		The factorization \(\sigma'=\phi_{L_YB_X}(\sigma)=\phi_{L_Y}\circ\phi_{B_X}(\sigma)\) from the alternative proof makes two simple special cases particularly transparent: \(Y=0\) (pure multiplicative update) and \(X=1\) (pure additive noise).

			\paragraph{Case \(Y=0\) (multiplicative update).}
			Here \(\sigma' = X \sigma X^{\dag}=\phi_{B_X}(\sigma)\), and \cref{eq:chanineq} becomes \(X \sigma_{z} X^{\dag} \ge \sigma_{z}\). The corresponding embedded matrix is \(\bar E=\bar B_X\coloneqq \Lambda B_X\Lambda^{-1}\), i.e.
		\begin{align} \label{eq:sympexp_X}
			\bar E_X & =
		\half
		\begin{pmatrix}
			\pqty{X + X^{- \dag}}
		 &
		\pqty{X - X^{- \dag}} \sigma_{x}
		\\
		\pqty{X^{*} - X^{- T}} \sigma_{x}
		 &
		\pqty{X^{*} + X^{- T}}
			\end{pmatrix}
		\end{align}
		(which is \cref{eq:sympexp} with \(Y=0\)).
	If in addition \(X\) is ABC-symplectic, \(X\coloneqq S\in\symcayn\), then the inequality saturates and \(\bar E_X\in\symcayN\). In this symplectic case one may use \(S^{-\dag} = \sigma_{z} S \sigma_{z}\) and write
	\begin{align}
		X \coloneqq S =
		\begin{pmatrix}
			\alpha    & \beta      \\
		\beta^{*} & \alpha^{*}
		\end{pmatrix}
		\,,
	\end{align}
			in which case \(\bar E_X=\Lambda B_S\Lambda^{-1}\) coincides with the previously-defined embedding \(S\mapsto S\bcc\), recovering \cref{thm:unitary-action-mixed} as the symplectic subcase of the channel embedding theorem.

			\paragraph{Case \(X=1\) (additive noise).}
				Here \(\sigma'=\sigma+Y=\phi_{L_Y}(\sigma)\), and the channel condition \cref{eq:chanineq} reduces to \(Y\ge 0\) (with \(Y=Y^{\dag}\)). In the disk picture the embedded matrix is simply
				\(\bar E=\bar L_Y\coloneqq \Lambda L_Y\Lambda^{-1}\), i.e.
			\begin{align}
				\bar L_Y
				=
				\begin{pmatrix}
				1+Y & -Y\sigma_x \\
				Y^{*}\sigma_x & 1-Y^{*}
			\end{pmatrix},
			\end{align}
			which is exactly \cref{eq:sympexp} with \(X=1\).

			Note that \(X\) need not be symplectic. Vacuum preservation also has a simple characterization in the disk picture. Since the vacuum corresponds to the origin \(\amat_{0}=0\), and \(\phi_{\bar E}(0)=CA^{-1}\) in the DCBA convention \cref{def:fractional-transformation}, a channel is vacuum-preserving if and only if the lower-left block of \(\bar E\) vanishes, i.e.\ \(\bar E\) is block upper-triangular. Substituting \(E_{21}=0\) into \cref{eq:sympexp} gives
	\begin{align}\label{eq:vacuum-invariance-xy}
		Y=\half\pqty{1-XX^{\dag}}.
	\end{align}
	In the unitary special case \(Y=0\), this reduces to \(XX^{\dag}=1\), i.e.\ the passive (\(U(n)\)) subgroup, for which \(\bar E\) is block diagonal.

\subsection{Graphical calculus}
\label{sec:graphical-calculus-sketch}
The graphical calculus for Gaussian pure states~\cite{menicucci_graphical_2011} can be summarized by a simple principle: a symmetric-matrix state space, together with explicit transformation rules, is already a graphical calculus. Indeed, both the Siegel disk and upper half-plane represent pure Gaussian states by symmetric matrices that can be read as weighted adjacency matrices, and Gaussian dynamics acts on these representatives by linear fractional transformations; this turns state updates into algebraic graph rewrite rules (with composition given by matrix multiplication; \cref{lem:fractional-composition-law}).

From this viewpoint, the main conceptual work needed to extend the graphical calculus beyond pure states is to show that mixed states admit a compatible symmetric-matrix representative \(\amat\) in an enlarged domain and that physical Gaussian dynamics still acts by closed-form fractional updates on those representatives. This is precisely what we establish here: mixed Gaussian states are identified as the physical subset \(\sspace\subset\siegeldiskN\), and deterministic Gaussian channels act on \(\sspace\) by \(\amat\mapsto\phi_{\bar E}(\amat)\) (Theorem~\ref{thm:channel-embedding}), extending the unitary rule \(K\mapsto\phi_S(K)\) from the pure-state disk~\cite{gabay_passive_2016}.

Finally, since the Cayley transform provides a biholomorphic isomorphism between disk and half-plane pictures (and extends directly to the doubled setting), the same mixed-state graphical calculus can be transported to the double upper half-plane as well, yielding a half-plane formulation that extends the pure-state graphical calculus of Ref.~\cite{menicucci_graphical_2011}. Developing a full diagrammatic language, simplification moves, and worked examples for mixed states and channels is substantial work, but the results above remove the main conceptual obstacle: the requisite state space and closed update rules are already in place.

\section{Conclusion and outlook}\label{sec:conclusion-outlook}
In conclusion, we have shown how to extend the description of Gaussian dynamics for pure states to include mixed states and deterministic Gaussian channels. The procedure consists of doubling the disk by embedding the pair \(\siegeldiskn\), \(\sympluscayn\) into \(\siegeldiskN\), \(\sympluscayN\). Since \(\siegeldiskN\) parametrizes Gaussian kernels (oscillator-semigroup elements) and contains the state space \(\sspace\) as a proper subset, it was necessary to identify the states within \(\siegeldiskN\). At the same time, much like \(\siegeldiskN\) is too large a space for Gaussian states, \(\sympluscayN\) is too large a semigroup for deterministic channels. We identified a simple subset of \(\sympluscayN\), corresponding to the normalized oscillator semigroup in the disk (matrix semigroup \(\sympluscayn\)), by parametrizing it directly in the disk. Its action on the state space \(\sspace\) gives an embedding of \(\sympluscayn\) into \(\sympluscayN\). We have also identified how Gaussian unitaries embed in the disk picture, \(\symcayn \subset \sympluscayn\). The entire set of deterministic channels is harder to identify within the disk framework, and in doing so we resorted to the traditional description of them in terms of affine transformations on the covariance matrix. This approach has the benefit that, given the \(X\)-\(Y\) parametrization of a deterministic channel, we can immediately write the corresponding fractional transformation. Although this characterization is complete, it would be useful to describe these as a more abstract, geometric subset of \(\sympluscayN\). We highlight this as an open problem beyond the scope of the present work, but note that a way to proceed in order to resolve the question would entail working exclusively in the larger disk, and asking that the Fock-Bargmann kernel of operations in \(\sympluscayN\) are deterministic channels. This approach is more cumbersome and does not say anything about the relationship between the disk and the covariance matrix description of Gaussian dynamics. 

Lastly, the fractional-update rules on disk representatives suggest an immediate graphical-calculus interpretation: read the symmetric matrices \(K\) (pure) and \(\amat\) (mixed) as weighted adjacency matrices and read the fractional maps as rewrite rules; see \cref{sec:graphical-calculus-sketch}. Developing a full diagrammatic calculus for mixed states and channels is an interesting direction for future work.

\acknowledgments
We thank Filippo Miatto and Rafael N.\ Alexander for helpful discussions.
The Commonwealth of Australia (represented by the Advanced Strategic
Capabilities Accelerator) supports this research through a Defence Science
Partnerships agreement. We acknowledge support from the Australian Research
Council Centre of Excellence for Quantum Computation and Communication
Technology (Project No. CE170100012). N.C.M.\ was supported by an ARC Future
Fellowship (Project No. FT230100571).

\bibliography{references}
\clearpage

\appendix
\onecolumngrid
\section{Auxiliary proofs and derivations}\label{sec:appendix}

\subsection{Proof of \cref{thm:sympplus-preserves-disk}}\label{app:thm-sympplus-preserves-disk}
We prove that the Siegel-disk preserving semigroup \(\sympluscayn\) preserves the Siegel disk \(\siegeldiskn\) under linear fractional transformations.

\begin{proof}
	The proof combines two stacked-notation invariants with two different involutions: symmetry uses the symplectic condition \(T^{T}\sigma_{y}T=\sigma_{y}\) (transpose and \(\sigma_y^{T}=-\sigma_y\)), while disk membership uses the semigroup inequality \(T^{\dag}\sigma_{z}T\ge\sigma_{z}\) (adjoint and \(\sigma_z=\sigma_z^{\dag}=\mathrm{diag}(1,-1)\)); see \cref{def:lifted-paulis}.

	Fix \(T\in\sympluscayn\) and write it in \(n\times n\) blocks as \(T=\begin{psmallmatrix}A&B\\ C&D\end{psmallmatrix}\). Let \(K\in\siegeldiskn\) and define the stacked vector
	\(
		X\coloneqq\begin{bmatrix}1\\ K\end{bmatrix}
		\),
		so that \(TX=\begin{bmatrix}A+BK\\ C+DK\end{bmatrix}\).

	\paragraph{Invertibility of the denominator.}
	For \(K\in\siegeldiskn\), the disk inequality \(K^{\dag}K<1\) is equivalent to
	\begin{align}\label{eq:diskpos-stacked}
		X^{\dag}\sigma_z X = 1-K^{\dag}K>0
		\,.
	\end{align}
	Since \(T\in\sympluscayn\), we have \(T^{\dag}\sigma_z T\ge\sigma_z\), and therefore
	\begin{align}\label{eq:disk-congruence-ineq}
		(TX)^{\dag}\sigma_z(TX)
		=
		X^{\dag}T^{\dag}\sigma_z T X
		\ge
		X^{\dag}\sigma_z X
		>
		0
		\,.
	\end{align}
	Suppose, for contradiction, that \(BK+A\) is not invertible. Then there exists \(v\neq 0\) such that \((BK+A)v=0\). Setting \(w\coloneqq Xv\neq 0\), we have
	\(
		Tw=(TX)v=\begin{bmatrix}0\\ (DK+C)v\end{bmatrix}
	\),
	so
	\(
	w^{\dag}T^{\dag}\sigma_z T w=(Tw)^{\dag}\sigma_z(Tw)=-\|(DK+C)v\|^{2}\le 0,
	\),
		where we used that \(\sigma_z=\begin{psmallmatrix}1&0\\ 0&-1\end{psmallmatrix}\) (with \(n\times n\) blocks) so that \(\begin{bmatrix}0\\ y\end{bmatrix}^{\dag}\sigma_z\begin{bmatrix}0\\ y\end{bmatrix}=-\|y\|^{2}\).
	contradicting \cref{eq:disk-congruence-ineq} together with \cref{eq:diskpos-stacked}. Hence \(BK+A\) is invertible and \(K'=\phi_T(K)\) is well-defined.

	\paragraph{Preservation of symmetry.}
	Symmetry of \(K\) is equivalent to the isotropy condition \cref{eq:siegelreq1}, i.e. \(X^{T}\sigma_y X=0\). Since \(T\in\symplecticcomplexn\), we have \(T^{T}\sigma_y T=\sigma_y\), hence
	\(
	(TX)^{T}\sigma_y(TX)=X^{T}T^{T}\sigma_y T X=X^{T}\sigma_y X=0
	\).
	With \(BK+A\) invertible, define \(K'\coloneqq\phi_T(K)=(DK+C)(BK+A)^{-1}\) and note the stacked identity
	\begin{align}\label{eq:stacked-action-disk}
			T\begin{bmatrix}1\\ K\end{bmatrix}
			=
			\begin{bmatrix}1\\ K'\end{bmatrix}(BK+A)
			\,,
		\end{align}
		which is obtained by multiplying out the block matrix. This implies that \(X'=\begin{bmatrix}1\\ K'\end{bmatrix}\) is also isotropic:
	\(
	0=(TX)^{T}\sigma_y(TX)=(BK+A)^{T}X'^{T}\sigma_y X'(BK+A)
	\),
	so \(X'^{T}\sigma_y X'=0\) and therefore \(K'^{T}=K'\).

	\paragraph{Preservation of \(K^{\dag}K<1\).}
		Define \(K'=\phi_T(K)\) and \(X'=\begin{bmatrix}1\\ K'\end{bmatrix}\). From \cref{eq:stacked-action-disk} we have \(TX=X'(BK+A)\). Taking the Hermitian quadratic form with \(\sigma_z\) gives
	\begin{align}\label{eq:diskposprime-stacked}
		(TX)^{\dag}\sigma_z(TX)
		=
		(BK+A)^{\dag}\,X'^{\dag}\sigma_z X'\,(BK+A),
	\end{align}
	and therefore
	\begin{align}\label{eq:diskposprime-congruence}
		X'^{\dag}\sigma_z X'
		=
		(BK+A)^{-\dag}X^{\dag}T^{\dag}\sigma_z T X(BK+A)^{-1}
		\ge
		(BK+A)^{-\dag}X^{\dag}\sigma_z X(BK+A)^{-1}
		>
		0,
	\end{align}
	where the inequality uses \(T^{\dag}\sigma_z T\ge\sigma_z\) and the strict positivity uses \cref{eq:diskpos-stacked}. By \cref{eq:diskpos-stacked} applied to \(K'\), the condition \(X'^{\dag}\sigma_z X'>0\) is exactly \(1-K'^{\dag}K'>0\), i.e. \(K'^{\dag}K'<1\).

	Therefore \(K'\in\siegeldiskn\). Since \(K\in\siegeldiskn\) was arbitrary, this shows \(\phi_T(\siegeldiskn)\subseteq\siegeldiskn\) for all \(T\in\sympluscayn\).
\end{proof}

\subsection{Proof of \cref{thm:symcayn-preserves-disk-transitive}}\label{app:thm-symcayn-preserves-disk-transitive}
We give a self-contained proof that the ABC symplectic group \(\symcayn\) preserves the Siegel disk \(\siegeldiskn\) and acts transitively on it. The main stacked-notation identities used below are sketched in the main text around \cref{eq:siegelreqproof1,eq:siegelreqproof2}.

\begin{proof}
\par\noindent\emph{Preservation of the disk.}
If \(S\in\symcayn\), then by \cref{eq:defsympcomp} we have \(S^{\dag}\sigma_{z}S=\sigma_{z}\), hence in particular \(S\in\sympluscayn\) (the semigroup inequality holds with equality). Therefore disk invariance is an immediate corollary of \cref{thm:sympplus-preserves-disk}.

\par\noindent\emph{Transitivity (explicit vacuum-to-\(K\) map).}
Fix \(K\in\siegeldiskn\). Write \(\siegeldiskn=\{K\in M_n(\complex)\,|\,K^{T}=K,\ K^{\dag}K<1\}\), so \(I-K^{\dag}K\) is positive definite. Define
\begin{align}
	\alpha \coloneqq (I-K^{\dag}K)^{-1/2},
	\qquad
	\beta^{*} \coloneqq K\alpha,
	\qquad
	\beta \coloneqq (\beta^{*})^{*}=K^{*}\alpha^{*},
\end{align}
and set
\begin{align}
	S_{K}\coloneqq
	\begin{pmatrix}
		\alpha & \beta \\
		\beta^{*} & \alpha^{*}
	\end{pmatrix}
	\,.
\end{align}
We claim \(S_{K}\in\symcayn\) and \(\phi_{S_{K}}(0)=K\).

First, \(\alpha\) is invertible and \(\beta^{*}\alpha^{-1}=K\), so \(\phi_{S_K}(0)=\beta^{*}\alpha^{-1}=K\) (using \cref{eq:linearfractional} at \(K=0\)).

Second, \(S_{K}\in\symcayn\): it is in ABC block form by construction, and the symplectic relations \cref{eq:alphabetaone,eq:alphabetatwo} hold. Indeed,
\begin{align}
	\alpha^{\dag}\alpha-\beta^{T}\beta^{*}
	&=
	\alpha^{\dag}\alpha-(K^{*}\alpha^{*})^{T}(K\alpha)
	=
	\alpha^{\dag}\alpha-\alpha^{\dag}(K^{*})^{T}K\alpha
	\\
	&=
	\alpha^{\dag}(I-K^{\dag}K)\alpha
	=
	I,
\end{align}
where we used \(K^{T}=K\) so \((K^{*})^{T}=K^{\dag}\), and \(\alpha=(I-K^{\dag}K)^{-1/2}\). Similarly,
\begin{align}
	\alpha^{T}\beta^{*}
	=
	\alpha^{T}K\alpha
	=
	(K^{*}\alpha^{*})^{\dag}\alpha
	=
	\beta^{\dag}\alpha,
\end{align}
where \(K^{\dag}=K^{*T}\) and \(K^{T}=K\) imply \((K^{*}\alpha^{*})^{\dag}=\alpha^{T}K\). Thus \(S_{K}\in\symcayn\) and reaches \(K\) from the vacuum.

\par\noindent\emph{Transitivity (arbitrary \(K_{1}\to K_{2}\)).}
Given \(K_{1},K_{2}\in\siegeldiskn\), choose \(S_{1},S_{2}\in\symcayn\) such that \(\phi_{S_{j}}(0)=K_{j}\) as above. Then \(S\coloneqq S_{2}S_{1}^{-1}\in\symcayn\) and, using the composition law \cref{eq:compositionlaw},
\begin{align}
	\phi_{S}(K_{1})
	=
	\phi_{S_{2}}\!\circ\phi_{S_{1}^{-1}}(K_{1})
	=
	\phi_{S_{2}}(0)
	=
	K_{2},
\end{align}
so the action is transitive.
\end{proof}
\subsection{Proof of \cref{thm:mixed-state-characterization}}\label{app:thm-mixed-state-characterization}

\begin{proof}
	\noindent\emph{Step 1 (block reduction of the UP).}
	Let \(\amat\in\gammaspace\) and write it in blocks as in \cref{eq:ablocks}. Set \(X\coloneqq\sigma_{x}\amat\). Since \(\amat\in\gammaspace\subset\siegeldiskN\), \(X\) is Hermitian with \(\operatorname{spec}(X)\subset(-1,1)\), hence \(B\coloneqq(1-X)^{1/2}\) exists, is Hermitian, and is invertible.

	The matrix on the left-hand side of \cref{eq:amatupshort} is Hermitian, so \cref{eq:amatupshort} is equivalent to
	\begin{align}\label{eq:upcong}
		B\pqty{(1+X)(1-X)^{-1}-\sigma_{z}}B \ge 0
		\,,
	\end{align}
	by congruence with the invertible matrix \(B\). Using that \(B\) commutes with \(X\) (functional calculus for Hermitian matrices) and that \(B^2=1-X\), we get
	\begin{align}
		B\pqty{(1+X)(1-X)^{-1}-\sigma_{z}}B
		&=
		(1+X)-B\sigma_{z}B
		\nonumber\\
		&=
		(1+X)-B(2P_{+}-1)B
		\nonumber\\
		&=
		2\pqty{1-BP_{+}B}
		\,,
	\end{align}
	where \(P_{+}\coloneqq(1+\sigma_{z})/2\) is the projector onto the upper block in the \(\sigma_{z}\) decomposition.
	Thus \cref{eq:upcong} is equivalent to \(1-BP_{+}B\ge0\).

	Set \(T\coloneqq P_{+}B\). Then \(BP_{+}B=T^{\dag}T\) and \(P_{+}B^{2}P_{+}=TT^{\dag}\), so \(BP_{+}B\) and \(P_{+}B^{2}P_{+}\) have the same nonzero eigenvalues. Since \(B^{2}=1-X\), we have
	\[
		P_{+}B^{2}P_{+}
		=
		P_{+}(1-X)P_{+}
		\,.
	\]
	In blocks, \(X=\sigma_{x}\amat=\begin{psmallmatrix}W^{*}&K^{*}\\ K&W\end{psmallmatrix}\), so \(P_{+}(1-X)P_{+}=1-W^{*}\).
	Therefore the eigenvalues of \(BP_{+}B\) are the eigenvalues of \(1-W^{*}\) together with zeros, and the eigenvalues of \(1-BP_{+}B\) are the eigenvalues of \(W^{*}\) together with ones.

	Finally, since \(X\) is Hermitian, its lower-right block \(W\) is Hermitian, so \(W^{*}\) and \(W\) have the same (real) spectrum.
	Hence \(1-BP_{+}B\ge0\) if and only if \(W\ge0\), which is equivalent to \cref{eq:amatupshort}.

	\par\noindent\emph{Step 2 (state characterization).}
	If \(\sigma\) is a physical covariance matrix, then \(\amat\in\sspace\) by \cref{thm:mixed-state-conditions-disk}. In particular, \(\amat\in\gammaspace\) and \cref{eq:amatupshort} holds, hence \(W\ge0\) by Step 1.

	Conversely, let \(\amat\in\gammaspace\) with \(W\ge0\). By Step 1, \(\amat\) satisfies \cref{eq:amatupshort}. Define
	\[
		\sigma \coloneqq \half(1+\sigma_{x}\amat)(1-\sigma_{x}\amat)^{-1}
		\,.
	\]
	Since \(\amat\in\siegeldiskN\) and \(\sigma_x\amat\) is Hermitian, we have \(1-\sigma_x\amat>0\), so \(\sigma\) is well-defined, Hermitian, and positive definite. Moreover, \cref{eq:amatupshort} is exactly \(\sigma-\half\sigma_{z}\ge0\), i.e. the UP. Thus \(\sigma\) satisfies the defining algebraic conditions of a (complex) Gaussian covariance matrix, and \(\amat\) represents a Gaussian state.
\end{proof}

\subsection{Proof of \cref{thm:unitary-action-mixed}}\label{app:thm-unitary-action-mixed}

We prove that the Gaussian unitary update \(\sigma'=S\sigma S^{\dag}\), with \(S\in\symcayn\), becomes the fractional transformation \(\amat'=\phi_{S\bcc}(\amat)\) on the adjacency matrix, and that this action preserves \(\gammaspace\) and \(\sspace\).

Throughout this proof, \(\sigma_{x}\) and \(\sigma_{z}\) denote the \(2n\times 2n\) lifted Pauli matrices unless stated otherwise; in the \(4n\times 4n\) context (for \(S\bcc\)) they denote \(\sigma_{x}\otimes I_{2n}\) and \(\sigma_{z}\otimes I_{2n}\), respectively. The notation follows the convention established in the main text.

\begin{proof}
\noindent\emph{Symplecticity identities.}
Since \(S\in\symcayn\), \cref{eq:defsympcomp} gives \(S\sigma_{z}S^{\dag}=\sigma_{z}\), hence \(S^{-\dag}=\sigma_{z}S\sigma_{z}\). Additionally, \cref{eq:defblockym} gives \(S\sigma_{x}=\sigma_{x}S^{*}\).

\par\noindent\emph{Coordinate transport.}
The adjacency matrix corresponding to \(\sigma\) is \(\amat=\sigma_{x}(\sigma-\tfrac12)(\sigma+\tfrac12)^{-1}\), with inverse
\begin{align}
	\sigma
	=
	\tfrac12
	(\sigma_{x}+\amat^{*})
	(\sigma_{x}-\amat^{*})^{-1}
	\,,
\end{align}
where we used the Hermiticity form \(\sigma=\sigma^{\dag}\) and the ABC property \(\amat\sigma_{x}=\sigma_{x}\amat^{*}\) (valid because \(\amat\in\gammaspace\)).

The invertibility of \((\sigma_{x}-\amat^{*})\) is guaranteed by \(\amat\in\siegeldiskN\): the disk condition \(\amat^{*}\amat<1\) implies \(\|\amat\|<1\) (operator norm), so \(\sigma_{x}-\amat^{*}=\sigma_{x}(1-\sigma_{x}\amat^{*})\) is invertible because \(\|\sigma_{x}\amat^{*}\|=\|\amat\|<1\).

\par\noindent\emph{Derivation of the fractional-transform formula.}
The updated adjacency matrix is
\begin{align}
	\amat'
	=
	\sigma_{x}
	(S\sigma S^{\dag}-\tfrac12)
	(S\sigma S^{\dag}+\tfrac12)^{-1}
	\,.
\end{align}
Using \(S^{-\dag}=\sigma_{z}S\sigma_{z}\), this becomes \cref{eq:atransf1}. Substituting \(\sigma=\tfrac12(\sigma_{x}+\amat^{*})(\sigma_{x}-\amat^{*})^{-1}\) into \cref{eq:atransf1} and simplifying:
\begin{align}
	\nonumber
	S\sigma -\tfrac12\sigma_{z}S\sigma_{z}
	&=
	\tfrac12\bigl[
		(S+\sigma_{z}S\sigma_{z})\amat^{*}
		+
		(S-\sigma_{z}S\sigma_{z})\sigma_{x}
	\bigr]
	(\sigma_{x}-\amat^{*})^{-1}
	\,,
\end{align}
and similarly for the denominator, yielding
\begin{align}
	\amat'
	&=
	\sigma_{x}
	\bigl[
		(S+\sigma_{z}S\sigma_{z})\amat^{*}
		+
		(S-\sigma_{z}S\sigma_{z})\sigma_{x}
	\bigr]
	\nonumber \\ &\quad \times
	\bigl[
		(S-\sigma_{z}S\sigma_{z})\amat^{*}
		+
		(S+\sigma_{z}S\sigma_{z})\sigma_{x}
	\bigr]^{-1}
	\,.
\end{align}
Now apply \(S\sigma_{x}=\sigma_{x}S^{*}\) and \(\amat\sigma_{x}=\sigma_{x}\amat^{*}\) (justified by \(S\in\symcayn\) and \(\amat\in\gammaspace\), respectively) to move all starred quantities to unstarred ones. Concretely,
\(\sigma_{x}\cdot S\cdot\amat^{*}=\sigma_{x}S\sigma_{x}\cdot\sigma_{x}\amat^{*}=S^{*}\cdot\amat\sigma_{x}\);
in the full numerator/denominator expression this rightmost \(\sigma_{x}\) is cancelled by the corresponding \(\sigma_{x}\) arising from the same manipulation on the denominator side (using \(\sigma_x^2=1\)). After commuting \(\sigma_{x}\) through, we obtain
\begin{align}
	\amat'
	&=
	\bigl[
		(S^{*}+\sigma_{z}S^{*}\sigma_{z})\amat
		+
		(S^{*}-\sigma_{z}S^{*}\sigma_{z})\sigma_{x}
	\bigr]
	\nonumber \\ &\quad \times
	\bigl[
		(S-\sigma_{z}S\sigma_{z})\sigma_{x}\amat
		+
		(S+\sigma_{z}S\sigma_{z})
	\bigr]^{-1}
	\,.
\end{align}
To identify this as \(\phi_{S\bcc}(\amat)\), write \(S=\alpha\pcc+\beta\pcc\sigma_{x}\) in \(n\times n\) blocks, and observe that \(\sigma_{z}\) commutes with the block-diagonal part and anticommutes with the off-diagonal part:
\begin{align}
	S+\sigma_{z}S\sigma_{z}=2\alpha\pcc
	\,,\qquad
	S-\sigma_{z}S\sigma_{z}=2\beta\pcc\sigma_{x}
	\,.
\end{align}
(Here we used \(\sigma_{z}(\alpha\pcc)\sigma_{z}=\alpha\pcc\) and \(\sigma_{z}(\beta\pcc\sigma_{x})\sigma_{z}=-\beta\pcc\sigma_{x}\), both of which follow from \(\sigma_{z}\sigma_{x}=-\sigma_{x}\sigma_{z}\).)
Similarly, \(S^{*}+\sigma_{z}S^{*}\sigma_{z}=2\alpha^{*}\pcc\) and \(S^{*}-\sigma_{z}S^{*}\sigma_{z}=2\beta^{*}\pcc\sigma_{x}\). Substituting and canceling factors of \(2\):
\begin{align}
		\amat'
		=
		((\alpha^{*}\pcc)\pcc\,\amat + (\beta^{*}\pcc)\pcc\,\sigma_{x})
		((\beta\pcc)\pcc\,\sigma_{x}\amat + (\alpha\pcc)\pcc)^{-1}
		=
		\phi_{S\bcc}(\amat)
		\,,
\end{align}
where the last equality uses \cref{eq:rep1} for the block structure of \(S\bcc\). This is \cref{eq:fractionalonmixed}.

\par\noindent\emph{Well-definedness (invertibility of the denominator).}
We must show that \((\beta\pcc)\pcc\,\sigma_{x}\amat+(\alpha\pcc)\pcc\) is invertible whenever \(\amat\in\siegeldiskN\) and \(S\in\symcayn\).

Since \(S\bcc\in\symcayN\) (the symplectic relations \cref{eq:alphabetaone,eq:alphabetatwo} lift block-by-block under \(\bullet\pcc\)), this follows from the standard invertibility theorem for fractional transformations on Siegel disks: if \(T\in\symcayn\) and \(K\in\siegeldiskn\), then the denominator of \(\phi_{T}(K)\) is invertible. Applied with \(n\mapsto 2n\), \(T=S\bcc\in\symcayN\), and \(K=\amat\in\siegeldiskN\), we obtain the result.

Alternatively, one can verify this directly: \cref{eq:alphabetaone} gives \(\alpha^{\dag}\alpha=1+\beta^{T}\beta^{*}\ge 1\), so \(\alpha\) is invertible with \(\|\alpha^{-1}\beta\|\le\|\alpha^{-1}\|\|\beta\|<1\). The same bound holds block-by-block after lifting, so \(\|(\alpha\pcc)\pcc^{-1}(\beta\pcc)\pcc\|<1\). Since \(\|\amat\|<1\), the denominator \((\alpha\pcc)\pcc(1+(\alpha\pcc)\pcc^{-1}(\beta\pcc)\pcc\sigma_{x}\amat)\) is invertible by a Neumann-series argument.

\par\noindent\emph{Preservation of \(\gammaspace\).}
We need \(\amat'\sigma_{x}=\sigma_{x}\amat'^{*}\), i.e., \(\amat'\) is ABC. This follows by transporting through covariance space. If \(\sigma\) is ABC (\(\sigma_{x}\sigma\sigma_{x}=\sigma^{*}\)), then using \(S\sigma_{x}=\sigma_{x}S^{*}\):
\begin{align}
	\sigma_{x}\sigma'\sigma_{x}
	=\sigma_{x}S\sigma S^{\dag}\sigma_{x}
	=S^{*}\sigma_{x}\sigma\sigma_{x}S^{T}
	=S^{*}\sigma^{*}S^{T}
	=(S\sigma S^{\dag})^{*}
	=\sigma'^{*}
	\,,
\end{align}
so \(\sigma'\) is ABC. Since the coordinate change \(\sigma\mapsto\amat(\sigma)\) preserves the ABC property (as shown in the derivation following \cref{eq:aofsigma}), \(\amat'=\amat(\sigma')\in\gammaspace\).

\par\noindent\emph{Preservation of \(\sspace\) (the uncertainty principle).}
We need \(\amat'\in\sspace\), i.e., the block \(W'\ge 0\) in \(\amat'=\begin{psmallmatrix}K'&W'\\ W'^{*}&K'^{*}\end{psmallmatrix}\).

This follows by transporting the uncertainty principle through the congruence. If \(\amat\in\sspace\), then \(\sigma\) satisfies \(\sigma-\tfrac12\sigma_{z}\ge 0\). Since \(S\sigma_{z}S^{\dag}=\sigma_{z}\) (\cref{eq:defsympcomp}):
\begin{align}
	\sigma'-\tfrac12\sigma_{z}
	=
	S\sigma S^{\dag}-\tfrac12 S\sigma_{z}S^{\dag}
	=
	S(\sigma-\tfrac12\sigma_{z})S^{\dag}
	\ge 0
	\,,
\end{align}
where positivity is preserved because \(X\mapsto SXS^{\dag}\) is a congruence (\(S\) is invertible). Therefore \(\sigma'\) satisfies the uncertainty principle, so \(\amat'\in\sspace\) by \cref{thm:mixed-state-conditions-disk}.

\par\noindent\emph{Projective ambiguity.}
Since the fractional transformation \(\phi_{S\bcc}\) depends only on the ratio of the blocks, \(\phi_{S\bcc}=\phi_{(-S)\bcc}\). Thus the induced action \(S\mapsto \phi_{S\bcc}\) is well-defined on the quotient \(\symcayn/\{\pm 1\}\). Independently, the Hilbert-space realization is only defined projectively (global phases/signs), and this projective ambiguity is invisible at the level of the fractional action on \(\amat\).
\end{proof}

\subsection{Proof of \cref{thm:channel-embedding}}\label{app:thm-channel-embedding}
We derive the explicit disk-picture channel update \(\amat'=\phi_{\bar E}(\amat)\) and the embedded matrix \(\bar E\) in \cref{eq:sympexp} by transporting the affine covariance update \(\sigma'=X\sigma X^{\dag}+Y\) through the coordinate change \(\sigma\leftrightarrow \amat\). This is algebraically analogous to the unitary case \cref{thm:unitary-action-mixed}, except that the channel has an affine term \(Y\).

\begin{proof}
\noindent\emph{Start from the coordinate update.}
By definition of the disk representative \(\amat(\sigma)\), the post-channel adjacency matrix is
\begin{align}
	\amat'
	&=
	\sigma_x
	\pqty{X \sigma X^{\dag} + Y - \half }
	\pqty{ X \sigma X^{\dag} + Y + \half }^{-1}
	\,.
\end{align}
Since \cref{eq:sympexp} involves \(X^{-T}\) and \(X^{-\dag}\), we assume throughout that \(X\) is invertible.

\par\noindent\emph{Substitute the inverse coordinate map.}
For \(\amat\in\gammaspace\) (so \(\amat\sigma_x=\sigma_x\amat^{*}\)), it is convenient to use the inverse of the coordinate map in the form
\begin{align}
	\label{eq:sigmainvAstar}
	\sigma
	=
	\tfrac12
	(\sigma_x+\amat^{*})
	(\sigma_x-\amat^{*})^{-1}
	\,.
\end{align}
Substituting \cref{eq:sigmainvAstar} into \(\sigma'=X\sigma X^{\dag}+Y\) and right-multiplying numerator and denominator by the invertible matrix \(X^{-\dag}(\sigma_x-\amat^{*})\) clears \((\sigma_x-\amat^{*})^{-1}\) and yields
\begin{align}
	\label{eq:channel-cayley-cleared}
	\amat'
	=
	\sigma_x M_{-}\,M_{+}^{-1},
	\qquad
	M_{\pm}\coloneqq X(\sigma_x+\amat^{*})+(2Y\pm 1)X^{-\dag}(\sigma_x-\amat^{*})
	\,.
\end{align}

\par\noindent\emph{Expand and factor the numerator.}
Using the ABC relations \(\sigma_x X=X^{*}\sigma_x\), \(\sigma_x Y=Y^{*}\sigma_x\), \(\sigma_x\amat^{*}=\amat\sigma_x\), and \(\sigma_x X^{-\dag}=X^{-T}\sigma_x\), we find
\begin{align}
	\sigma_x M_{-}
	&=
	X^{*}(1+\amat\sigma_x)+(2Y^{*}-1)X^{-T}(1-\amat\sigma_x)
	\nonumber\\
	&=
	\Bqty{X^{*}-X^{-T}+2Y^{*}X^{-T}}
	+
	\Bqty{X^{*}+X^{-T}-2Y^{*}X^{-T}}\,\amat\sigma_x
	\,,
\end{align}
and hence, right-factoring \(\sigma_x\) (so the \(\sigma_x\) is moved to the right rather than dropped),
\begin{align}
	\label{eq:channel-num-factor}
	\sigma_x M_{-}
	=
	\Bqty{
		\Bqty{X^{*}+X^{-T}-2Y^{*}X^{-T}}\,\amat
		+
		\Bqty{X^{*}-X^{-T}+2Y^{*}X^{-T}}\,\sigma_x
	}\sigma_x
	\,.
\end{align}

\par\noindent\emph{Expand and factor the denominator.}
Similarly, expanding \(M_{+}\) and using \(\amat^{*}=\sigma_x\amat\sigma_x\) gives
\begin{align}
	\label{eq:channel-den-factor}
	M_{+}\sigma_x
	&=
	\Bqty{X-X^{-\dag}-2Y X^{-\dag}}\,\sigma_x\amat
	+
	\Bqty{X+X^{-\dag}+2Y X^{-\dag}}
	\,.
\end{align}

\par\noindent\emph{Identify the fractional transformation.}
Combining \cref{eq:channel-num-factor} with \((M_{+}\sigma_x)^{-1}=\sigma_x M_{+}^{-1}\), we obtain
\begin{align}
	\amat'
	&=
	\pqty{
		\begin{aligned}
			&\pqty{X^{*} + X^{- T} - 2 Y^{*} X^{- T}} \amat
			\\
			&\quad+
			\pqty{X^{*} - X^{- T} + 2 Y^{*} X^{- T}} \sigma_{x}
		\end{aligned}
	}
	\pqty{
		\begin{aligned}
			&\pqty{X - X^{- \dag} - 2 Y X^{- \dag}} \sigma_{x} \amat
			\\
			&\quad+
			\pqty{X + X^{- \dag} + 2 Y X^{- \dag}}
		\end{aligned}
	}^{- 1}
	\,.
\end{align}
This is exactly the linear fractional transformation \(\amat'=\phi_{\bar E}(\amat)\) in our DCBA convention \cref{def:fractional-transformation}, with \(\bar E\) as in \cref{eq:sympexp}.
\end{proof}

\end{document}